\documentclass[a4paper, 10pt]{cmslatex}

\usepackage{amsmath, amssymb, amsfonts}
\usepackage{multirow}
\usepackage{booktabs}
\usepackage{appendix}
\usepackage{graphicx, xcolor}
\usepackage{overpic,subfig, float}
\usepackage{listings, color}
\newcommand{\be}{\begin{equation}}
\newcommand{\ee}{\end{equation}}
\newcommand{\ba}{\begin{array}}
\newcommand{\ea}{\end{array}}
\newcommand{\bea}{\begin{eqnarray}}
\newcommand{\eea}{\end{eqnarray}}
\newcommand{\beaa}{\begin{eqnarray*}}
\newcommand{\eeaa}{\end{eqnarray*}}

\newcommand\+[1]{\boldsymbol{#1}}
\newcommand\bbR{\mathbb{R}}
\newcommand\bbN{\mathbb{N}}
\newcommand\Kn{\mathit{Kn}}

\newcommand\vf[1]{\bar{f}_{\boldsymbol{e}_1+ {#1} \boldsymbol{e}_2}}
\newcommand\vfb[1]{\bar{f}_{\boldsymbol{e}_1+ ({#1}) \boldsymbol{e}_2}}

\newcommand\pd[2]{\dfrac{\partial {#1}}{\partial {#2}}}
\newcommand\pdd[1]{\dfrac{\partial}{\partial {#1}}}
\newcommand\odd[1]{\dfrac{\mathrm{d}}{\mathrm{d} {#1}}}
\newcommand\od[2]{\dfrac{\mathrm{d} {#1}}{\mathrm{d} {#2}}}

\newcommand\Gauss[3]{\frac{1}{(2\pi {#3})^{3/2}}
\exp\left(-\frac{|{#1}-{#2}|^2}{2 {#3}}\right)}

\newcommand\ang[1]{\langle{#1}\rangle}
\newcommand\wang[2]{\langle\phi_{#1},\ \xi_2\varphi_{#2}\rangle_{\bar{\omega}}}

\begin{document}

\title{An Approximate Analytical Solution to Knudsen Layers}

\author{Ruo Li\thanks{CAPT, LMAM \& School of Mathematical Sciences,
    Peking University, Beijing, China, email: {\tt
      rli@math.pku.edu.cn}.} \and Yichen Yang\thanks{School of
    Mathematical Sciences, Peking University, Beijing, China, email:
    {\tt yichenyang@pku.edu.cn}.}
}

\maketitle
\begin{abstract}
  We apply moment methods to obtaining an approximate analytical solution
  to Knudsen layers. Based on the hyperbolic regularized moment
  system for the Boltzmann equation with the Shakhov collision model,
  we derive a linearized hyperbolic moment system to model the
  scenario with the Knudsen layer vicinity to a solid wall with
  Maxwell boundary condition. We find that the reduced system is in an
  even-odd parity form that the reduced system proves to be
  well-posed under all accommodation coefficients. We show that the
  system may capture the temperature jump coefficient and the thermal
  Knudsen layer well with only a few moments. With the increasing number
  of moments used, qualitative convergence of the approximate solution is observed.
\end{abstract}

%%% Local Variables:
%%% mode: latex
%%% TeX-master: "article"
%%% End:

\section{Introduction}\label{Sec.1}

The Knudsen layer is an important feature of the rarefied gas flow
\cite{Lilley2007}, where the continuum assumption does not hold so the
Navier-Stokes-Fourier (NSF) equations fail to describe the gas
behavior \cite{shen2006rare} but the model from a statistical
viewpoint such as the Boltzmann equation \cite{Boltzmann} works. As
introduced in the book \cite{Struchtrup}, the moment equations, which
are extend macroscopic transport equations reduced from the Boltzmann
equation, provide a new description of rarefied gases.  The model
reduction methods are necessary partly because the direct simulation
of the Boltzmann equation, such as the direct simulation Monte Carlo
(DSMC) \cite{Bird} and the discrete velocity method (DVM)
\cite{Broadwell}, may be too expensive for applications in concern
\cite{Reese2003, Mizzi2007}.

This paper is aimed to obtain an approximate analytical solution to 
the Knudsen layer in some classical flow problems, based on the 
hyperbolic regularized moment equations (HME) developed in recent
years \cite{Fan,Fan_new,ANRxx,framework,fan2016model}.
Moment methods for the Boltzmann equation are first proposed by Grad
\cite{Grad} and success to simulate the nonequilibrium gas flow with
high accuracy and high efficiency \cite{Muller, Struchtrup2002,
TorrilhonEditorial}. Nevertheless, the original Grad's moment equations 
suffer the lack of hyperbolicity \cite{Grad13toR13} and the hyperbolic 
model reduction remains an important issue in this area, whose long 
and rich history can be found in the review paper \cite{cai2020hyperbolic}. 
Following the regularization framework \cite{framework}, the HME is globally
hyperbolic regularized from the Grad's moment system of arbitrary
moment orders and has been studied in both theoretical and numerical
aspects \cite{Cai2011}.

There have been exhaustive studies applying the linearized Boltzmann
equation \cite{Williams2001} with various collision models to
Knudsen layers in classical flow problems, i.e. the temperature jump
problem \cite{Welander1954} and Kramers' problem \cite{Kramers1949}.
Many highly accurate numerical results have been reported
\cite{Loyalka1971, 1978Temperature, Siewert2000} by the
discrete-ordinates method. However, moment methods may bring different
insights into the understanding of Knudsen layers, especially by their
available analytical solutions. To our best knowledge,
\cite{2008Linear} first analyses a 1D linear kinetic equation for heat
transfer by means of Grad's moment methods.  \cite{Gu2014} presents
analytical solutions of the temperature jump problem for linearized
R13 and R26 moment methods. In the recent work \cite{Lijun2017}, formal
analytical solutions of the Kramers' problem are obtained for the
linearized HME with the BGK collision model \cite{BGK}.

In this paper, we rewrite the linearized HME as an even-odd parity
form and present approximate analytical solutions of the temperature
jump problem with the Shakhov collision model \cite{Shakhov}. Compared
to the early work \cite{Lijun2017} on Kramers' problem, the even-odd
parity form of the linearized HME is explicitly utilized in this paper
and the moment equations' boundary conditions are also imposed as an
even-odd formulation. In this way, we improve the results in
\cite{Lijun2017} that the well-posedness is attained under all
accommodation coefficients. Furthermore, the numerical study
confirms the effectivity of our model. The idea of the even-odd
formulation is inspired by analysis of the kinetic equations such as
\cite{egg2012, Lu2014}.

Briefly, we first derive the HME from the Boltzmann equation then make
linearization to get the linearized HME (LHME). Thanks to the
assumptions of the temperature jump problem, we can decouple the
equations including the temperature from the whole LHME, to get a
system of linear ordinary differential equations (ODEs) with constant
coefficients. We impose the boundary conditions of the moment system
by multiplying Maxwell's accommodation boundary condition
\cite{Maxwell} with some appropriate polynomials then integrating both
sides. Finally, we separate the decaying and non-decaying unknowns,
seeking the analytical solutions of the ODEs satisfying the boundary
conditions and the boundedness of the decaying unknowns.  For arbitrary
moment order $M$, the explicit solutions can be determined via a
simple algorithm and we then study the temperature jump coefficient,
temperature defect, and effective thermal conductivity, etc.

This paper is organized as follows. In Section 2 we derive the LHME
in the half-space with wall boundary conditions, and discuss its
reduced version in some classical flow problems. In Section 3 we
detailedly discuss the temperature jump problem, obtaining the
analytical solutions, and proving the well-posedness of the reduced
moment system. In Section 4 we briefly discuss the Kramers' problem.
In Section 5 we carefully compare our temperature profile with other
kinetic models both theoretically and numerically. The paper ends
with a conclusion.

%%% Local Variables:
%%% mode: latex
%%% TeX-master: "article"
%%% End:

\section{The Linearized Moment System}\label{Sec.2}
\subsection{The Basic Equations.} We consider the following
Boltzmann equation \cite{Boltzmann} with 
the Shakhov \cite{Shakhov} collision model 
\begin{eqnarray}\label{eq:Bol}
	\pd{f}{t} + \boldsymbol{\xi}\cdot\nabla_{\boldsymbol{x}}f = Q^S(f),
\end{eqnarray}
where $f=f(t,\+x,\+\xi)$ is the number density distribution function
of particles at time $t\in\bbR^+$,
location $\+x=(x_1,x_2,x_3)\in\bbR^3,$ with velocity
$\+\xi=(\xi_1,\xi_2,\xi_3)\in\bbR^3$. 
The Shakhov collision term $Q^S(f)$ is 
\begin{eqnarray}\label{eq:shakhov}
	Q^{S}(f) = \frac{1}{\tau}(f^{S}-f),\ 
f^S = \rho\omega^{[\boldsymbol{u},\theta]}(\boldsymbol{\xi})
\left(1+\frac{(1-\Pr)(\boldsymbol{\xi}-\boldsymbol{u})\cdot \boldsymbol{q}}{5\rho\theta^2}
\left(\frac{|\boldsymbol{\xi}-\boldsymbol{u}|^2}{\theta}-5\right)
\right),
\end{eqnarray}
where $\tau^{-1}$ measures frequency of the collision, $\Pr$ is the
Prandtl number. And 
the macroscopic variables such as density $\rho=\rho(t,\boldsymbol{x})$,
macro velocity vector $\+u=\boldsymbol{u}(t,\boldsymbol{x})$,
temperature $\theta=\theta(t,\boldsymbol{x})$ and 
heat flux vector $\+q=\boldsymbol{q}(t,\boldsymbol{x})$ are defined
by the distribution $f$:
\begin{eqnarray}\label{eq:macrodef}
\rho = \int_{\mathbb{R}^3}\!\!\!f\, \mathrm{d}\boldsymbol{\xi},\ 
\rho \boldsymbol{u} = \int_{\mathbb{R}^3}\!\!\!f \boldsymbol{\xi}\,\mathrm{d}\boldsymbol{\xi},\ 
\rho|\boldsymbol{u}|^2+3\rho\theta = \int_{\mathbb{R}^3}\!\!f|\boldsymbol{\xi}|^2\,\mathrm{d}\boldsymbol{\xi},\
\boldsymbol{q}=\frac{1}{2}\int_{\mathbb{R}^3}f|\boldsymbol{\xi}-\boldsymbol{u}|^2(\boldsymbol{\xi}-\boldsymbol{u})\,\mathrm{d}\boldsymbol{\xi}.
\end{eqnarray}
$\omega^{[\boldsymbol{u},\theta]}(\boldsymbol{\xi})$ is the local
Maxwellian, defined as 
\begin{eqnarray*}
\omega^{[\boldsymbol{u},\theta]}(\boldsymbol{\xi}) = \Gauss{\boldsymbol{\xi}}{\boldsymbol{u}}{\theta}.
\end{eqnarray*}
\vspace{-2mm}
\begin{remark}
The Shakhov model, which in fact turns to the BGK model \cite{BGK}
	when $\Pr=1$, may provide the correct Prandtl number of the flow.
 The Boltzmann operator with a more general collision kernel will be
	discussed in the future work but not in this paper. 
\end{remark}

Then we briefly introduce the deduction of the HME, whose more details can 
be found in \cite{framework}. First we make the Hermite expansion ansatz
of the distribution function
\begin{eqnarray} \label{eq:expan}
f(t,\boldsymbol{x},\boldsymbol{\xi}) = \omega ^{[\boldsymbol{u},\theta]}(\boldsymbol{\xi})
	\sum_{\boldsymbol{\alpha}\in \mathbb{N}^3}f_{\boldsymbol{\alpha}}(t,\boldsymbol{x})
	\mathrm{He}_{\boldsymbol{\alpha}}^{[\boldsymbol{u},\theta]}(\boldsymbol{\xi}),
\end{eqnarray}
where $\boldsymbol{\alpha}:=(\alpha_1,\alpha_2,\alpha_3)\in \mathbb{N}^3$.
$\mathrm{He}_{\+\alpha} ^{[\boldsymbol{u},\theta]}(\boldsymbol{\xi})$
is the generalized 3D Hermite polynomial defined as 
\begin{eqnarray*}
\mathrm{He}_{\boldsymbol{\alpha}}^{[\boldsymbol{u},\theta]}(\boldsymbol{\xi}) &=& 
	\frac{(-1)^{|\boldsymbol{\alpha}|}}
	{\omega ^{[\boldsymbol{u},\theta]}(\boldsymbol{\xi})}
	\dfrac{\partial^{|\boldsymbol{\alpha}|} \omega ^{[\boldsymbol{u},\theta]}(\boldsymbol{\xi}) }
	{\partial \boldsymbol{\xi}^{\boldsymbol{\alpha}}},\ 
|\boldsymbol{\alpha}|:=\alpha_1+\alpha_2+\alpha_3.
\end{eqnarray*}
By Appendix \ref{app:A}, $\{\mathrm{He}_{\+\alpha} ^{[\boldsymbol{u},\theta]}(\boldsymbol{\xi})\}$ are 
orthogonal polynomials with the weight function $\omega ^{[\boldsymbol{u},\theta]}(\boldsymbol{\xi})$ 
, and
\begin{eqnarray}\label{eq:falpha}
	f_{\boldsymbol{\alpha}} = \frac{\theta^{|\+\alpha|}}{\boldsymbol{\alpha}!}\int_{\bbR^3}\!\!\!
f \mathrm{He}_{\boldsymbol{\alpha}}^{[\boldsymbol{u},\theta]}(\boldsymbol{\xi})\, \mathrm{d} \boldsymbol{\xi},
\ \boldsymbol{\alpha}!:=\alpha_1!\alpha_2!\alpha_3!.
\end{eqnarray}
If $\boldsymbol{e}_i\in \mathbb{R}^3$ is the unit vector with
the $i$-th component equaling one, from \eqref{eq:macrodef} we have
\begin{eqnarray}
\label{eq:f123}
f_{\boldsymbol{0}}=\rho,\ f_{\boldsymbol{e}_i}=0,\ \sum_{d=1}^3f_{2\boldsymbol{e}_d}=0,\ 
	2f_{3 \boldsymbol{e}_i} + \sum_{d=1}^3f_{\boldsymbol{e}_i+2 \boldsymbol{e}_d} = q_i.
\end{eqnarray}

For any integer $M\geq2$, we define a 
projection $\mathcal{P}_M^{[\+u,\theta]}$ onto the space
spanned by the first $M$-th order basis functions
$ \mathcal{S}_M^{[\+u,\theta]}:= \mathrm{span}\{\omega ^{[\boldsymbol{u},\theta]}(\boldsymbol{\xi}) \mathrm{He}_{\+\alpha} ^{[\boldsymbol{u},\theta]}(\boldsymbol{\xi}),\ |\+\alpha|\leq M\}$, as 
\begin{eqnarray*}
	\mathcal{P}_M^{[\+u,\theta]} f := \omega ^{[\boldsymbol{u},\theta]}(\boldsymbol{\xi})
	\sum_{|\boldsymbol{\alpha}|\leq M}f_{\boldsymbol{\alpha}}(t,\boldsymbol{x})
	\mathrm{He}_{\boldsymbol{\alpha}}^{[\boldsymbol{u},\theta]}(\boldsymbol{\xi}).\quad
\end{eqnarray*}
Then substituting it into the Boltzmann equation, making the HME's
closure and matching the coefficients before the first $M$-th order
basis functions, we have a closed system with finite
terms which is called the $M$-th order HME:
\begin{eqnarray} \label{eq:cHME}
	\mathcal{P}_M^{[\+u,\theta]}\left(
\pd{\mathcal{P}_M^{[\+u,\theta]}f}{t} + 
	\sum_{i=1}^3\xi_i {\mathcal{P}_M^{[\+u,\theta]}}\left(
\pd{\mathcal{P}_M^{[\+u,\theta]}f}{x_i}
	\right)\right) = \mathcal{P}_M^{[\+u,\theta]} Q^S(\mathcal{P}_M^{[\+u,\theta]}f).
\end{eqnarray}
The component form reads as
\begin{eqnarray}\label{eq:tHME} 
	\begin{aligned} \frac{\mathrm{D} f_{\boldsymbol{\alpha}}}{\mathrm{D} t} + 
		\sum_{d=1}^3 &\left(\theta\pd{f_{\boldsymbol{\alpha}- \boldsymbol{e}_d}}{x_d} + 
		(1 - \delta_{|\boldsymbol{\alpha}|,M})(\alpha_d + 1)
		\pd{f_{\boldsymbol{\alpha}+ \boldsymbol{e}_d}}{x_d} \right) \\  
		+ \sum_{k=1}^3 f_{\boldsymbol{\alpha}- \boldsymbol{e}_k}\frac{\mathrm{D} u_k}{\mathrm{D} t} 
    + \sum_{d=1}^3\sum_{k=1}^3 \pd{u_k}{x_d}
		&\left(\theta f_{\boldsymbol{\alpha}- \boldsymbol{e}_k- \boldsymbol{e}_d} + 
		(1 - \delta_{|\boldsymbol{\alpha}|,M})(\alpha_d + 1) f_{\boldsymbol{\alpha}- \boldsymbol{e}_k+\boldsymbol{e}_d}\right) \\
		+ \frac{1}{2} \sum_{k=1}^3 f_{\boldsymbol{\alpha}- 2 \boldsymbol{e}_k} \frac{\mathrm{D} \theta}{\mathrm{D} t}
    + \dfrac{1}{2} \sum_{d=1}^3\sum_{k=1}^3  \pd{\theta}{x_d}
		&\left(\theta f_{\boldsymbol{\alpha}-2 \boldsymbol{e}_k- \boldsymbol{e}_d} + 
		(1 - \delta_{|\boldsymbol{\alpha}|,M})(\alpha_d + 1)
		f_{\boldsymbol{\alpha}-2 \boldsymbol{e}_k+ \boldsymbol{e}_d}\right)
    = -Q_{\boldsymbol{\alpha}},
  \end{aligned}
\end{eqnarray}
where $\displaystyle\frac{\mathrm{D}}{\mathrm{D}t}=\pdd{t}+\sum_{d=1}^3u_d\pdd{x_d}$ is the material derivative and  
$(\cdot)_{\boldsymbol{\alpha}}$ is taken as zero if any component of $\boldsymbol{\alpha}$ is 
negative or $|\boldsymbol{\alpha}|>M$. $Q_{\+\alpha}$ is calculated directly as in \cite{Cai2011},
\begin{eqnarray*}
Q_{\+\alpha} = 
-\frac{\theta^{|\boldsymbol{\alpha}|}}{\boldsymbol{\alpha}!}\int_{\bbR^3}\!\!\!
Q^S(f)
\mathrm{He}_{\boldsymbol{\alpha}}^{[\boldsymbol{u},\theta]}(\boldsymbol{\xi})\, \mathrm{d} \boldsymbol{\xi}
=
\frac{1}{\tau}\left(\delta_{|\+\alpha|\geq2}f_{\+\alpha}-
\frac{1-\Pr}{5}\sum_{i,j=1}^3\delta_{\+\alpha,\+e_i+2\+e_j}q_i 
\right),
\end{eqnarray*}
where the Kronecker function $\delta_{\+\alpha,\+e_i+2\+e_j}$ equals 1 when $\+\alpha=\+e_i+2\+e_j$ and 
equals 0 otherwise; $\delta_{|\+\alpha|\geq2}$ equals 1 when $|\+\alpha|\geq2$ and equals 0 otherwise. 
\begin{remark}
	The hyperbolicity of the moment system \eqref{eq:tHME} is proved
	in \cite{framework}. And \cite{framework} clarifies that the
	keypoint of the regularization lies in the extra projection 
	after the space derivatives in \eqref{eq:cHME}, which is also the
	only difference between Grad's moment equations and the HME.
\end{remark}

After the non-dimensionalization and linearization, we will get the
linearized HME (LHME). The linearization is assumed to be 
around the Maxwellian $\rho_0\omega^{[\+u_0,\theta_0]}(\+\xi)$ where
$\rho_0,\+u_0=\+0,\theta_0$ are constants. 
Denote by $L$ a characteristic length, then we introduce the 
dimensionless coordinates $\bar{\+x}$ and time $\bar{t}$ 
as $\+x=L\bar{\+x}$ and $t=\frac{L}{\sqrt{\theta_0}}\bar{t}$. 
The corresponding dimensionless Knudsen number is defined as 
\begin{eqnarray*}
\Kn = \frac{\tau_0}{L/\sqrt{\theta_0}},\quad \text{where}\ \tau=\tau_0(1+\bar{\tau}),\ \bar{\tau}=o(1).
\end{eqnarray*}
Analogously, we introduce the variables with a bar as dimensionless
variables:
\begin{eqnarray}\label{eq:dimensionless}
\rho=\rho_0(1+\bar{\rho}),\ \boldsymbol{u}=\sqrt{\theta_0}\bar{\boldsymbol{u}},\ 
\theta=\theta_0(1+\bar{\theta}),\ 
f_{\boldsymbol{\alpha}} = \rho_0\theta_0^{\frac{|\boldsymbol{\alpha}|}{2}}
	\bar{f}_{\boldsymbol{\alpha}},\ 
	|\boldsymbol{\alpha}|\geq 2,
\end{eqnarray}
where $\Kn$ is assumed to be a small quantity, $\bar{\rho},\bar{\+u},
\bar{\theta}$ and $\bar{f}_{\+\alpha}$
assumed to be $O(\Kn)$. Substituting \eqref{eq:dimensionless} into
the HME \eqref{eq:tHME} and discarding
the higher order terms, we have $M$-th order LHME:
\begin{eqnarray} \label{eq:simHME}
	\pd{\bar{h}_{\+\alpha}}{\bar{t}}+
\sum_{d=1}^3\left(
	\pd{\bar{h}_{\boldsymbol{\alpha}- \boldsymbol{e}_d}}{\bar{x}_d} +
	(1-\delta_{|\boldsymbol{\alpha}|,M})(\alpha_d+1)\pd{\bar{h}_{\boldsymbol{\alpha}+ 
	\boldsymbol{e}_d}}{\bar{x}_d} 
\right) = -\frac{1}{\mathit{Kn}}\bar{Q}_{\boldsymbol{\alpha}},\ |\+\alpha|\leq M,
\end{eqnarray}
where $\bar{h}_{\+\alpha}$ and $\bar{Q}_{\+\alpha}$ are defined as
\begin{eqnarray}\label{eq:hQ}
\bar{h}_{\+\alpha} = \bar{f}_{\boldsymbol{\alpha}}+\sum_{k=1}^3\delta_{\boldsymbol{\alpha},\boldsymbol{e}_k}
\bar{u}_k + \frac{1}{2}\sum_{k=1}^3\delta_{\boldsymbol{\alpha},2 \boldsymbol{e}_k}\bar{\theta},\quad
\bar{Q}_{\boldsymbol{\alpha}} = \delta_{|\boldsymbol{\alpha}|\geq 2}\bar{f}_{\boldsymbol{\alpha}} -
\frac{1-\Pr}{5}\sum_{i,j=1}^3\delta_{\boldsymbol{\alpha},\boldsymbol{e}_i+2 \boldsymbol{e}_j}\bar{q}_i.
\end{eqnarray}
\begin{remark}\label{rem:Grad}
We can see from the deduction that the extra projection in the HME
	only affects the higher order terms. So the difference between Grad's
	moment equations and the HME vanishes in case of this linearization.
\end{remark}
\begin{remark}
From another point of view, the linearized moment equations may be
	deduced directly from the linearized Boltzmann equation by the
	traditional Galerkin spectral expansion, i.e. under the basis
	functions independent on temporal and spatial variables $t,\ x$. 
\end{remark}

\subsection{Wall Boundary Conditions.}
In this paper we will consider the half-space problem, where the gas 
flow is on the upper half plane of an
infinite plate wall. Without loss of generality, we assume the
coordinate of the wall $\+x=(x_1,0,x_3),x_1,x_3\in\bbR,$ the outer
normal vector $\+n=(0,-1,0)^T$, the wall velocity $\+u^W=(0,0,0)^T$
and the wall temperature $\theta^W$ here and hereafter.

We use Maxwell's accommodation boundary condition \cite{Maxwell} to
describe the diffuse-specular process between the wall and the gas flow,
which in this case reads as
\begin{eqnarray}
\label{eq:Maxwell-bc}
	f(t,\boldsymbol{x},\boldsymbol{\xi}) = 
	\chi \mathcal{M}^W(\boldsymbol{x},\boldsymbol{\xi}) + 
	(1-\chi)f(t,\boldsymbol{x},\boldsymbol{\xi}^*),&
	\quad \xi_2 > 0,\ x_2=0,
\end{eqnarray}
where $\chi\in[0,\ 1]$ is the accommodation coefficient. 
$\boldsymbol{\xi}^*=\boldsymbol{\xi}-2 \boldsymbol{n}(\boldsymbol{\xi}\cdot \boldsymbol{n}) 
= (\xi_1,-\xi_2,\xi_3)^T$ comes from the specular reflection at the wall,
and $\mathcal{M}^W(\+x,\+\xi)$ is the Maxwellian characterizing the wall:
\begin{eqnarray*}
\mathcal{M}^W(\boldsymbol{x}, \boldsymbol{\xi}) = \frac{\rho^W(\boldsymbol{x})}{(2\pi\theta^W(\boldsymbol{x}))^{\frac{3}{2}}}
  \mathrm{exp}\left( - \frac{|\boldsymbol{\xi}- \boldsymbol{u}^W|^2}{2\theta^W(\boldsymbol{x})}\right),
\end{eqnarray*}
where $\rho^W$ is a normalizing factor to ensure $(\+u-\+u^W)\cdot\+n=0$
at the wall.

To construct the boundary conditions of the moment equations, a
traditional way is multiplying \eqref{eq:Maxwell-bc} by some 
polynomials $p_{\+\alpha}(\+\xi)$ and taking integral in $\bbR^3$
about $\+\xi$ both sides. To ensure the continuity when $\chi\rightarrow 0$,
Grad \cite{Grad,Grad1958} suggests choosing the polynomials satisfying
\begin{eqnarray} \label{eq:oddpol}
p_{\boldsymbol{\alpha}}(\boldsymbol{\xi}) = -p_{\boldsymbol{\alpha}}
	(\boldsymbol{\xi}^*),\quad \deg(p_{\boldsymbol{\alpha}})\leq M,
\end{eqnarray}
i.e. the odd polynomials about $\xi_2$. Hence 
substituting $\+\xi^*$ by $\+\xi$ in the last integral and noting that
$\mathcal{M}^W(\boldsymbol{x}, \boldsymbol{\xi^*}) 
= \mathcal{M}^W(\boldsymbol{x}, \boldsymbol{\xi})$, we have
\begin{eqnarray}\label{eq:int-bc}
\int_{\bbR^3}\!\! p_{\boldsymbol{\alpha}}(\boldsymbol{\xi})f\, \mathrm{d} \boldsymbol{\xi} &=& 
\int_{\bbR^2}\!\!\int_{-\infty}^0\!\! 
	p_{\boldsymbol{\alpha}}(\boldsymbol{\xi})f\, \mathrm{d} \boldsymbol{\xi} + 
\int_{\bbR^2}\!\!\int_{0}^{+\infty}\!\! p_{\boldsymbol{\alpha}}(\boldsymbol{\xi})
\left(\chi \mathcal{M}^W+(1-\chi)f(\boldsymbol{x}, \boldsymbol{\xi}^*)\right)
\, \mathrm{d} \boldsymbol{\xi} \notag \\ 
& = & \chi
\int_{\bbR^2}\!\!\int_{-\infty}^0\!\! p_{\boldsymbol{\alpha}}(\boldsymbol{\xi})
\left(f- \mathcal{M}^W\right)\, \mathrm{d} \boldsymbol{\xi}.
\end{eqnarray}

Since the equivalence of finite dimensional polynomial spaces, the specific
choice of $p_{\+\alpha}(\+\xi)$ can somehow be arbitrary in numeric,
such as the Hermite polynomial \cite{Cai2011}, the Legendre
polynomial \cite{2008Linear}, or even the monomial
$\boldsymbol{\xi}^{\boldsymbol{\alpha}}$ \cite{Lijun2017}. However, to
analyse the well-posedness, it may be more convenient to rewrite the
odd polynomial $p_{\+\alpha}$ as $\xi_2\tilde{p}_{\+\alpha}$ where
$\tilde{p}_{\+\alpha}$ is an even polynomial. This is analogous to the
Marshak conditions which impose the continuity of fluxes in the 
domain decomposition methods \cite{1997Arnold} and its benefits
will show naturally in the following sections.

According to this belief and note that when $\alpha_2$ is even,
$\mathrm{He}_{\+\alpha} ^{[\boldsymbol{u},\theta]}(\boldsymbol{\xi}) 
= \mathrm{He}_{\+\alpha}^{[\+u,\theta]}(\+\xi^*),$ we choose
\begin{eqnarray} \label{eq:phi}
p_{\boldsymbol{\alpha}}(\boldsymbol{\xi}) 
	= \xi_2\theta^{\frac{|\+\alpha|}{2}} \mathrm{He}_{\+\alpha} ^{[\boldsymbol{u},\theta]}(\boldsymbol{\xi}),\
\boldsymbol{\alpha} \in\mathbb{I} = \{|\boldsymbol{\alpha}|\leq M-1\ |\ \alpha_2\ \text{is even}\}.
\end{eqnarray}
Then assume $f$ and $\mathcal{M}^W$ each has the expansion
coefficients $f_{\+\alpha}$ and $m_{\+\alpha}$
defined as \eqref{eq:falpha} under the basis functions 
$\{\omega ^{[\boldsymbol{u},\theta]}(\boldsymbol{\xi}) 
\mathrm{He}_{\+\alpha} ^{[\boldsymbol{u},\theta]}(\boldsymbol{\xi})\}$.
The even-odd symmetry shows that
\beaa
\int_{\bbR^2}\int_{-\infty}^0\!\!\! p_{\+\alpha} \mathrm{He}_{\+\beta} 
^{[\boldsymbol{u},\theta]} \omega ^{[\boldsymbol{u},\theta]}
 \, \mathrm{d}\+\xi &= &
\int_{\bbR^2}\int^{+\infty}_0\!\!\! p_{\+\alpha} \mathrm{He}_{\+\beta} 
^{[\boldsymbol{u},\theta]} \omega ^{[\boldsymbol{u},\theta]}
 \, \mathrm{d}\+\xi \\ &=& \frac{1}{2}
\int_{\bbR^3}\!\!\! p_{\+\alpha} \mathrm{He}_{\+\beta} 
^{[\boldsymbol{u},\theta]} \omega ^{[\boldsymbol{u},\theta]}
 \, \mathrm{d}\+\xi,\ \beta_2\text{ is odd.}\\
-\int_{\bbR^2}\int_{-\infty}^0\!\!\! p_{\+\alpha} \mathrm{He}_{\+\beta} 
^{[\boldsymbol{u},\theta]} \omega ^{[\boldsymbol{u},\theta]}
 \, \mathrm{d}\+\xi &=& 
\int_{\bbR^2}\int^{+\infty}_0\!\!\! p_{\+\alpha} \mathrm{He}_{\+\beta} 
^{[\boldsymbol{u},\theta]} \omega ^{[\boldsymbol{u},\theta]}
 \, \mathrm{d}\+\xi \\&=& \frac{1}{2}
\int_{\bbR^3}\!\!\! |\xi_2|\theta^{\frac{|\+\alpha|}{2}}
\mathrm{He}_{\+\alpha} ^{[\boldsymbol{u},\theta]} \mathrm{He}_{\+\beta} 
^{[\boldsymbol{u},\theta]} \omega ^{[\boldsymbol{u},\theta]}
 \, \mathrm{d}\+\xi,\ \beta_2\text{ is even.}
\eeaa
Note that all the integral can calculate separately about
$\xi_1,\xi_2$ and $\xi_3$. After some tedious computation using the
properties of Hermite polynomials (Appendix \ref{app:A}), 
and plugging the linearization \eqref{eq:dimensionless} as well
as $\bar{m}_{\+\alpha}$ defined analogously (explicitly calculated
in Appendix \ref{app:B}), making the linearized HME's closure
$\bar{f}_{\+\alpha}=o(\Kn),\ |\+\alpha|>M$, then discarding all the
higher order small quantities, we have the linearized boundary
conditions from \eqref{eq:int-bc}:
\begin{eqnarray}
\label{eq:lbc-0317}
\alpha_2 !\left(\bar{f}_{\boldsymbol{\alpha}- \boldsymbol{e}_2} + 
(\alpha_2+1)\bar{f}_{\boldsymbol{\alpha}+ \boldsymbol{e}_2}\right)
=
b(\chi)\left(
\sum_{\beta_2=0,\ \text{even}}^{M-\alpha_1-\alpha_3}S(\alpha_2,\beta_2)
\left(\bar{f}_{k_{\beta_2}^{\boldsymbol{\alpha}}} -\bar{m}_{k_{\beta_2}^{\boldsymbol{\alpha}}}\right)\right),
\end{eqnarray}
where $\boldsymbol{\alpha}\in\mathbb{I} = \{|\boldsymbol{\alpha}|\leq M-1\ |\ \alpha_2\ \text{is even}\}$ as 
in \eqref{eq:phi}, 
$k^{\boldsymbol{\alpha}}_{\beta_2}:=\boldsymbol{\alpha}+(\beta_2-\alpha_2) \boldsymbol{e}_2$, 
$\displaystyle b(\chi) := \frac{2\chi}{2-\chi}(2\pi)^{-\frac{1}{2}}.$ $S(\alpha_2,\beta_2)$ is a
1D half-space integral with two parameters $\alpha_2,\beta_2\in\bbN$, defined as
\begin{eqnarray}
\label{eq:defS}
S(\alpha_2,\beta_2) := \sqrt{\frac{2\pi}{\theta}}\int_{-\infty}^0\!\!
\xi_2\theta^{\frac{\alpha_2+\beta_2}{2}} \mathrm{He}_{\beta_2}^{[0,\theta]}(\xi_2)
\mathrm{He}_{\alpha_2}^{[0,\theta]}(\xi_2)\omega^{[0,\theta]}(\xi_2) \mathrm{d}\xi_2.
\end{eqnarray}
We put all the calculation in Appendix for brevity and just
list some properties for completeness:
\begin{prop}
$S(\alpha_2,\beta_2)$ is independent of $\theta$ and can write explicitly.
	Especially, when $\alpha_2$ is even, $\beta_2$ is odd
	and $|\beta_2-\alpha_2|\neq 1$, we have
	$S(\alpha_2,\beta_2)=0.$ (proof in Appendix \ref{app:A})
\end{prop}
\begin{prop}
	When $\+\alpha\neq\+0,\+e_i,2\+e_i(1\leq i\leq3)$, $\bar{m}_{\+\alpha}=0.$ Especially,
	$\bar{m}_{\+e_i}=\bar{u}_i^W-\bar{u}_i,$ and 
	$\bar{m}_{2\+e_i}=\frac{1}{2}(\bar{\theta}^W-\bar{\theta}).$
	(proof in Appendix \ref{app:B})
\end{prop}

\begin{remark}
	The boundary conditions \eqref{eq:lbc-0317} are in 
	an even-odd parity form, i.e, the left-hand side of \eqref{eq:lbc-0317} 
	only involves $\bar{f}_{\+\alpha}$
	where $\alpha_2$ is odd and the right-hand side $\alpha_2$ is even.
\end{remark}

\subsection{Reduced Moment System.}\label{sec:25}
Under the assumptions of the temperature jump problem proposed by
Welander \cite{Welander1954}, we claim that the equations
including $\bar{\theta}$ can decouple from the whole LHME. Thus, we
only need to solve a smaller moment system, which is called
the reduced moment system, to get solutions of the temperature jump problem.

In the temperature jump problem, we assume that the gas velocity is
$\+u=(u_1,0,0)^T$ and all derivatives in $x_1,x_3,t$ vanish. Further,
we assume that there is a given constant gradient of the temperature
normal to the wall at infinity. Thus, we just
set $\boldsymbol{\alpha}=2 \boldsymbol{e}_k+i \boldsymbol{e}_2,\ k=1,2,3,\ 
0\leq i\leq M-2$ in the $M$-th order LHME \eqref{eq:simHME}, and 
set $\+\alpha=2\+e_k+i\+e_2,\ k=1,2,3,\ 0\leq i\leq M-3,\ i$ even,
in the linearized wall boundary conditions \eqref{eq:lbc-0317}.
Since $\bar{h}_{\+e_2}=\bar{u}_2=0$, \eqref{eq:f123} and
\eqref{eq:hQ}, we will have $3(M-1)$ equations with the same number
unknowns $\bar{h}_{\+\alpha}$, where
$\sum_{i=1}^3\bar{h}_{2\+e_i}=\frac{3}{2}\bar{\theta}$.
If we impose the remaining required boundary conditions by the
boundedness of solutions, we will get a system of ODEs with the
correct number of boundary conditions.

The details will be shown in the next section, and here we just
mention two important tricks. First, the main focus in the temperature
jump problem, i.e. $\bar{\theta}$, is only dependent on $x_2$, so we
can add the corresponding terms on $x_1$ and $x_3$ to get a reduced system
of $2(M-1)$ equations. This is similar as integrating in the $x_1$ and
$x_3$ dimension when applying the linearized Boltzmann equation to the
temperature jump problem \cite{Williams2001}. Second, since the
Boltzmann collision operator always has the nontrivial null space which
means the conservation laws, we can only expect part of
$\bar{f}_{\+\alpha}$ to be bounded at infinity. So we dividedly
consider what we call the decaying variables and non-decaying variables.
\begin{remark}
	The Kramers' problem \cite{Kramers1949}, which can be seen as the
	velocity analogue of the temperature jump problem, would also be
	solved by a reduced moment system. For LHME with the BGK collision
	model, this is claimed in \cite{Lijun2017}. Here for the
	Shakhov collision model, $\bar{q}_1=3\bar{f}_{3\+e_1}+
	\bar{f}_{\+e_1+2\+e_2}+\bar{f}_{\+e_1+2\+e_3}$ will appear.
	So similarly we can set  $\+\alpha=\boldsymbol{e}_1+ 
	i \boldsymbol{e}_2,\  0\leq i \leq M-1$ and assume
	$\bar{f}_{3\+e_1}=0$, $\bar{f}_{\+e_1+2\+e_3}=0$ to get $M$
	equations with the same number unknowns.
\end{remark}

\section{The Temperature Jump Problem}\label{sec.3}
%\subsection{The Basic Equations.}\label{sec:31}

For simplicity, we write $x_2$ as $y$ and define 
$\bar{\omega}(\+\xi) = \omega^{[\+0,1]}(\+\xi)$,
\begin{eqnarray} \label{eq:simfg}
	\bar{g}_i=\bar{f}_{2\boldsymbol{e}_1+i \boldsymbol{e}_2}+\bar{f}_{2\boldsymbol{e}_3+i \boldsymbol{e}_2},
	\ \bar{t}_i=\bar{f}_{(i+2) \boldsymbol{e}_2},\ \overline{\mathrm{He}}_{\+\alpha}(\+\xi)= 
	\mathrm{He}^{[\+0,1]}_{\+\alpha}(\+\xi), \  
	\ang{\cdot}_{\bar{\omega}}=\displaystyle\int_{\bbR^3}\!\!\cdot\bar{\omega}\, \mathrm{d}\+\xi.
\end{eqnarray}
\begin{theorem}
	Then if $\bar{\theta}$ satisfies \eqref{eq:simHME},
	it must satisfy the following $2(M-1)$ equations:
\begin{eqnarray} \label{eq:bartheta}
\od{\bar{q}_2}{\bar{y}} = 0,\quad \od{\bar{\theta}}{\bar{y}} &=& 
	-\frac{2}{5}\frac{1}{\mathit{Kn}}\Pr\bar{q}_2 - \frac{4}{5}\od{(\bar{t}_0+(1-\delta_{M,3})
	(6\bar{t}_2+\bar{g}_2))}{\bar{y}},\\
\label{eq:mat-temp}
\boldsymbol{M}\od{\hat{w}}{\bar{y}} &:=& \begin{bmatrix}
	\boldsymbol{0} & \boldsymbol{M}_0 \\
	\boldsymbol{M}_0^T & \boldsymbol{0}
\end{bmatrix}\odd{\bar{y}} \begin{bmatrix} \hat{w} _{\text{even}} \\ \hat{w} _{\text{odd}}
\end{bmatrix}
	= -\frac{1}{\mathit{Kn}}\hat{w},
\end{eqnarray}
where $M\geq3$, $\bar{q}_2=3\bar{t}_1+\bar{g}_1$, 
	$\hat{w}=\+L\hat{f}:=(\hat{w} _{\text{even}},\hat{w} _{\text{odd}})^T$
with $\+L= \mathrm{diag}(\+L_1,\+L_2)$, $\hat{f}=(\hat{f} _{\text{even}}, \hat{f} _{\text{odd}})^T$.
Here $\hat{f} _{\text{even}}=(\bar{t}_0,\bar{t}_2,\bar{g}_2,\cdots,\bar{t}_{m_e-1},\bar{g}_{m_e-1})^T
\in \mathbb{R}^{m_e}$ 
collects unknowns with even subscripts, $\hat{f} _{\text{odd}}=
(\bar{t}_1-\frac{\bar{q}_2}{5},\bar{t}_3,\bar{g}_3,\cdots,\bar{t}_{m_o},\bar{g}_{m_o})^T\in \mathbb{R}^{m_o}$ 
collects the odd. The index $m_o = 2\lfloor \frac{M-1}{2}\rfloor-1,\quad m_e = 2\lfloor \frac{M}{2}\rfloor-1,$
thus $m_o+m_e=2(M-2)$. The matrix $\+M_0=(m_{ij}^0)\in \mathbb{R}^{m_e\times m_o},\+L_1= \mathrm{diag}(a_i)
\in \mathbb{R}^{m_e\times m_e},\+L_2 = \mathrm{diag}(b_i)\in \mathbb{R}^{m_o\times m_o}$ can write explicitly:
\begin{eqnarray} \label{eq:m12t}
	m_{ij}^0 = \frac{1}{a_ib_j}\langle \phi_i, \xi_2\varphi_j \rangle_{\bar{\omega}},\ 
	a_i := \sqrt{\langle \phi_i, \phi_i\rangle_{\bar{\omega}}};\ 
	b_i := \sqrt{\langle \varphi_i, \varphi_i\rangle_{\bar{\omega}}},
\end{eqnarray}
where $\phi=(\phi_i)\in \mathbb{R}^{m_e}$ and $\varphi=(\varphi_j)\in \mathbb{R}^{m_o}$
come from rearranging the Hermite polynomials:
\begin{equation*}
\phi_1 = \overline{\mathrm{He}}_{2 \boldsymbol{e}_2}-\frac{1}{2}\left(
\overline{\mathrm{He}}_{2 \boldsymbol{e}_1}+\overline{\mathrm{He}}_{2 \boldsymbol{e}_3} \right),
\phi_{2k}= \overline{\mathrm{He}}_{(2k+2) \boldsymbol{e}_2},\ \phi_{2k+1}
= \frac{1}{2}\left(\overline{\mathrm{He}}_{2 \boldsymbol{e}_1+ 2k \boldsymbol{e}_2}+
\overline{\mathrm{He}}_{2 \boldsymbol{e}_3+ 2k \boldsymbol{e}_2}\right);
\end{equation*}
\vspace{-4mm}
\begin{equation*}
\small
\varphi_1 = \overline{\mathrm{He}}_{3 \boldsymbol{e}_2}-\frac{3}{2}\left(
\overline{\mathrm{He}}_{2 \boldsymbol{e}_1+ \boldsymbol{e}_2}+\overline{\mathrm{He}}_{2 \boldsymbol{e}_3+ \boldsymbol{e}_2} \right),
\varphi_{2k} = \overline{\mathrm{He}}_{(2k+3) \boldsymbol{e}_2},\ \varphi_{2k+1}
= \frac{1}{2}\left(\overline{\mathrm{He}}_{2 \boldsymbol{e}_1+ (2k+1) \boldsymbol{e}_2}+
\overline{\mathrm{He}}_{2 \boldsymbol{e}_3+ (2k+1) \boldsymbol{e}_2}\right).
\end{equation*}
\end{theorem}
\begin{proof}
As mentioned before, we set $\boldsymbol{\alpha}=2 \boldsymbol{e}_k+i \boldsymbol{e}_2,\ k=1,2,3,\ 
0\leq i\leq M-2,\ M\geq3$ in \eqref{eq:simHME}. For each $i$, add the equations with $k=1$ and $k=3$,
then we have
\beaa
	\odd{\bar{y}}(\bar{g}_{i-1}+\delta_{i,1}\bar{\theta}) + (1-\delta_{i,M-2})(i+1)\odd{\bar{y}}\bar{g}_{i+1}
	= - \frac{1}{\Kn}\left(\bar{g}_i-2\frac{1-\Pr}{5}\delta_{i,1}\bar{q}_2\right),
\eeaa
	where $\bar{g}_{-1}=0$.
And for $k=2$ we have
\beaa
	\odd{\bar{y}}\left(\bar{t}_{i-1}+\frac{1}{2}\delta_{i,1}\bar{\theta}\right) + 
	(1-\delta_{i,M-2})(i+3)\odd{\bar{y}}\bar{t}_{i+1} = -\frac{1}{\Kn}\left(
	\bar{t}_i-\frac{1-\Pr}{5}\delta_{i,1}\bar{q}_2\right),
\eeaa
	where $\bar{t}_{-1}=0$ because $\bar{f}_{\+e_2}=\bar{u}_2=0$ in this case. Note that
	$\bar{g}_0+\bar{t}_0=0$ and $\bar{q}_2=3\bar{t}_1+\bar{g}_1$ from \eqref{eq:f123}, so
	add the equations with $i=0$, we have
	\beaa
	\od{\bar{q}_2}{\bar{y}} = 0\qquad\Rightarrow\qquad \bar{q}_2 \text{ is a constant.}
	\eeaa
	Add the equations with $i=1$ by proportion 3:1 to make the right hand side a constant,
	we have
	\beaa
\od{\bar{\theta}}{\bar{y}} &=& 
	-\frac{2}{5}\frac{1}{\mathit{Kn}}\Pr\bar{q}_2 - \frac{4}{5}\od{(\bar{t}_0+(1-\delta_{M,3})
	(6\bar{t}_2+\bar{g}_2))}{\bar{y}}.	
	\eeaa
	Note Appendix \ref{app:A} tells us 
	$\xi_d \overline{\mathrm{He}}_{\+\alpha} = \alpha_d \overline{\mathrm{He}}_{\+\alpha-\+e_d} + \overline{
	\mathrm{He}}_{\+\alpha+\+e_d},\ \ang{\overline{\mathrm{He}}_{\+\alpha},
	\overline{\mathrm{He}}_{\+\beta}}_{\bar{\omega}}
	= {\+\alpha!}\ \delta_{\+\alpha,\+\beta}$, so direct computation shows that from \eqref{eq:m12t},
	\beaa
	a_1=\sqrt{3},\ a_{2k}=\sqrt{(2k+2)!},\ a_{2k+1}=\sqrt{(2k)!}\ ,\\
	b_1=\sqrt{15},\ b_{2k}=\sqrt{(2k+3)!},\ b_{2k+1}=\sqrt{(2k+1)!}\ ,\\
	\wang{1}{1}=9,\ \wang{2}{1}=24,\ \wang{3}{1}=-6, \\
	\wang{2k}{2k}=(2k+3)!,\quad \wang{2k}{2k-2}=(2k+2)!,\\ 
	\wang{2k+1}{2k+1}=(2k+1)!,\quad \wang{2k+1}{2k-1}=(2k)!.
	\eeaa
	And the other entries of $\+M_0$ are zero. So if we eliminate $\bar{\theta}$ in the later equations,
	we can verify that $\bar{\theta}$ satisfies \eqref{eq:bartheta}\eqref{eq:mat-temp}.
\hfill
\end{proof}

Since the special structure of $\+M$ from the even-odd parity form, immediately we have
\begin{lemma}\label{lem:01}
	$\boldsymbol{M}$ has $m_o$ positive , $m_o$ negative and $m_e-m_o$ zero eigenvalues.
\end{lemma}
\begin{proof}
Evidently, if $\begin{bmatrix} \boldsymbol{0} & \boldsymbol{M}_0 \\
	\boldsymbol{M}_0^T & \boldsymbol{0} \end{bmatrix} \begin{bmatrix} x_1 \\ x_2 \end{bmatrix}
		= \lambda \begin{bmatrix} x_1 \\ x_2 \end{bmatrix},\ \text{then} \
\begin{bmatrix} \boldsymbol{0} & \boldsymbol{M}_0 \\
	\boldsymbol{M}_0^T & \boldsymbol{0} \end{bmatrix} \begin{bmatrix} x_1 \\ -x_2 \end{bmatrix}
		= -\lambda \begin{bmatrix} x_1 \\ -x_2 \end{bmatrix}.$ So the positive and negative 
			eigenvalues of $\+M$ appear in pairs. $\+M$ is real symmetric so it can be real 
			diagonalized. What's more, we claim that $\+M_0$
			has a column full rank of $m_o$. If suppose the contrary,
			there exists non-trivial coefficients
			$r_j\in\bbR,1\leq j\leq m_o$ such that
	\vspace{-3mm}{\setlength{\belowdisplayskip}{3pt}
	 \begin{eqnarray*}
	\ang{\phi_i,\xi_2\sum_{j=1}^{m_o}r_j\varphi_j}_{\bar{\omega}}=0,\quad \forall 1\leq i\leq m_e.
	\end{eqnarray*} }
But if we let $i=1$, the trinomial recurrence and orthogonality tell us 
	\vspace{-3mm}{\setlength{\belowdisplayskip}{3pt}
	 \begin{eqnarray*}
	\ang{\phi_1,\xi_2\sum_{j=1}^{m_o}r_j\varphi_j}_{\bar{\omega}}=
	\ang{\phi_1,\xi_2r_1\varphi_1}_{\bar{\omega}}= 9r_1,
	\end{eqnarray*} }
so $r_1=0$. By induction, we can see that $r_j=0,1\leq j\leq m_o$. Thus the lemma is proved.
\hfill
\end{proof}

From the process of the proof we immediately have
\begin{cor} There exist an orthogonal diagonalization
	$\boldsymbol{M} \boldsymbol{R}= \boldsymbol{R} \boldsymbol{\Lambda}$ 
	where $\boldsymbol{R}^T \boldsymbol{R}= \boldsymbol{R} \boldsymbol{R}^T = 
	\boldsymbol{I}_{2(M-2)}$ is the $2(M-2)$-th order identity matrix and
\begin{eqnarray}
\label{eq:Rexp}
\boldsymbol{R}:=
\begin{bmatrix}
\boldsymbol{R} _{\text{even}} & \boldsymbol{R}_0 & \boldsymbol{R} _{\text{even}} \\
\boldsymbol{R} _{\text{odd}} & \boldsymbol{0} & - \boldsymbol{R} _{\text{odd}}
\end{bmatrix},\quad
	\boldsymbol{\Lambda} := \begin{bmatrix} \boldsymbol{\Lambda}_+ & & \\
	& \boldsymbol{0}_{m_e-m_o} & \\
		& & - \boldsymbol{\Lambda}_+
	\end{bmatrix},
\end{eqnarray}
where $\boldsymbol{R} _{\text{even}}\in \mathbb{R}^{m_e\times m_o},\ \boldsymbol{R} _{\text{odd}}\in \boldsymbol{R}^{m_o\times m_o},\ 
	\boldsymbol{R}_0\in \boldsymbol{R}^{m_e\times(m_e-m_o)}$,
	$\boldsymbol{\Lambda}_+:= \mathrm{diag}(\lambda_i)\in \mathbb{R}^{m_o\times m_o}$ 
	and $\lambda_1\geq\lambda_2\geq\cdots\geq\lambda_{m_o}>0$. 
\end{cor}

%\subsection{Boundary Conditions.}
As mentioned in Section \ref{sec:25}, we shall impose boundary
	conditions from two parts: one is to ensure the boundedness of the
	decaying variables $\hat{w}$, the other is from the wall-gas
	interaction. Define $\hat{v}=\+R^{-1}\hat{w}$,
	we have characteristic equations
\begin{eqnarray} \label{eq:cha}
	\boldsymbol{\Lambda}\od{\hat{v}}{\bar{y}} = 
	-\frac{1}{\mathit{Kn}}\hat{v}.
\end{eqnarray}
	Assume $\hat{v}=(\hat{v}_+,\hat{v}_0,\hat{v}_-)^T$ where $\hat{v}_+,\hat{v}_-\in\bbR^{m_o},\hat{v}_0\in
	\bbR^{m_e-m_o}$. Then if we don't allow the exponential blow up of $\hat{v}$ at infinity, 
	we would get $m_e$ boundary conditions:
\begin{eqnarray}
\label{eq:bc1}
	\hat{v}_-(0) = \+0,\quad \hat{v}_0(0)=\+0.
\end{eqnarray}
As a remark, from $\hat{w}=\boldsymbol{R} \hat{v}$ and \eqref{eq:bc1} we have
\begin{eqnarray}
\label{eq:col1}
	\hat{w} _{\text{even}} = \boldsymbol{R} _{\text{even}} \hat{v}_+,\quad 
	\hat{w} _{\text{odd}} = \boldsymbol{R} _{\text{odd}} \hat{v}_+.
\end{eqnarray}
\begin{theorem}
	From \eqref{eq:lbc-0317}, $\hat{v}_+(0)$ and $\bar{\theta}(0)$ satisfy 
	the following $m_o+1$ boundary conditions:
\begin{eqnarray}
\label{eq:nwbc-2}
	\bar{q}_2c_r + \+E\begin{bmatrix} 0 & \\ &
		\boldsymbol{M}_0
\end{bmatrix} 
	\begin{bmatrix} \bar{\theta}(0)-\bar{\theta}^W \\ 
		\hat{w}_{\text{odd}}(0)
	\end{bmatrix} = b(\chi) \+E\boldsymbol{T}\begin{bmatrix} \bar{\theta}(0)-\bar{\theta}^W \\
		\hat{w}_{\text{even}}(0)
\end{bmatrix},
\end{eqnarray}
	where $c_r=(1,\frac{4}{5\sqrt{3}},\frac{2\sqrt{6}}{5},\frac{2\sqrt{2}}{5},0,...,0)^T
	\in \mathbb{R}^{m_o+1}$, $\+E:=\begin{bmatrix}\+I_{m_o+1},\+0 \end{bmatrix}\in\bbR^{(m_o+1)\times (m_e+1)}$.
		$\hat{w} _{\text{even}} , \hat{w} _{\text{odd}}$ are related to $\hat{v}_+$ by \eqref{eq:col1}. And
	\beaa
	\+T = \begin{bmatrix} 1 & \\ & \+L_1^{-1} \end{bmatrix}
		\begin{bmatrix} \+P_1 & \\ & \+I_{m_e-1} \end{bmatrix} \+T^b
\begin{bmatrix} \+P_1 & \\ & \+I_{m_e-1} \end{bmatrix} 
	\begin{bmatrix} 1 & \\ & \+L_1^{-1} \end{bmatrix}, \quad 
\+P_1=\begin{bmatrix} 0.5 & 1 \\ 1 & -1\end{bmatrix}. 
	\eeaa
Here $\boldsymbol{T}^b:=(t^b_{ij})\in \mathbb{R}^{(m_e+1)\times(m_e+1)}$ is a square matrix 
	with nonzero entries 
	\vspace{-3mm}
\begin{eqnarray}\label{eq:tb}
t^b_{2k,\ 2l} = S(2k-2,\ 2l-2)\ ,\ t^b_{2k-1,\ 2l-1} = S(2k,\ 2l)-\frac{S(2k,\ 0)}{S(0,\ 0)}S(0,\ 2l),
\ k,l\geq 1,
\end{eqnarray}
where $S(\alpha,\beta)$ is the half-space integral \eqref{eq:defS}. 
\end{theorem}
\begin{proof}
We just set $\boldsymbol{\alpha}=2 \boldsymbol{e}_k+i \boldsymbol{e}_2$, where $i$ is even 
and $0\leq i\leq M-3$, in the linearized boundary conditions \eqref{eq:lbc-0317}.
For each $i$, add the conditions with $k=1$ and $k=3$, then we will have
	$2\lfloor \frac{M-1}{2}\rfloor=m_o+1$ more boundary conditions. Since $\bar{t}_0+\bar{g}_0=0$, we have
\beaa
	\+P_1\begin{bmatrix} \bar{\theta} \\ \bar{t}_0 \end{bmatrix} = 
		\begin{bmatrix} \bar{t}_0+\frac{1}{2}\bar{\theta} \\ \bar{g}_0+\bar{\theta}
			\end{bmatrix}.
\eeaa
	From Appendix \ref{app:B} we know that $\bar{m}_{2\+e_k}=\frac{1}{2}(\bar{\theta}^W-\bar{\theta})$,
	$\bar{m}_{\+\alpha}=0$ when $\alpha_2$ is even and $|\+\alpha|>2$, and
	\beaa
	S(0,0)(\bar{m}_{\+0}-\bar{\rho}) = \sum_{\beta_2=2,\text{ even}}^MS(0,\beta_2)(\bar{f}_{\beta_2 \+e_2}-
	\bar{m}_{\beta_2\+e_2}).
	\eeaa
	So direct computation will show \eqref{eq:nwbc-2} right.
\hfill
\end{proof}
\begin{remark}
Here when $M$ is odd, $m_e=m_o=M-2$. So $\+M$ has no zero eigenvalues and 
	$\+E=\+I_{m_o+1}$ is an identity matrix which can be ignored. When $M$ is even, $m_e=M-1,\ m_o=M-3$. So
	$\+M$ has two zero eigenvalues and $\+E$ ensures the correct number of boundary conditions. 
\end{remark}
\begin{lemma}\label{lem:02}
	$\boldsymbol{T}^b$ is negative symmetric definite, so is $\+T$.
\end{lemma}
\begin{proof}
	By definition $S(\alpha,\beta)=S(\beta,\alpha)$, so $\boldsymbol{T}^b$ is symmetric. 
	For any $x=(x_i)\in \mathbb{R}^{m_e+1}$, define
\begin{eqnarray*}
	f_o(\xi_2) := \sum_{i=1}^{(m_e+1)/2}x_{2i-1} \overline{\mathrm{He}}_{2i}(\xi_2),\quad
	f_e(\xi_2) := \sum_{i=1}^{(m_e+1)/2}x_{2i} \overline{\mathrm{He}}_{2i-2}(\xi_2),
\end{eqnarray*}
	where $\overline{\mathrm{He}}_{\alpha}(\xi_2):= \mathrm{He}_{\alpha}^{[0,1]}(\xi_2)$. 
	By definition \eqref{eq:defS} and Cauchy-Schwartz inequality, we have 
\begin{eqnarray*}
	x^T \boldsymbol{T}^b x &=& \sqrt{\frac{2\pi}{\theta}}\int_{-\infty}^0\!\!\xi_2\left(f_e^2 + f_o^2
	\right)\omega^{[0,\theta]}\mathrm{d}\xi_2 - \frac{1}{S(0,0)}
	\frac{2\pi}{\theta}\left(\int_{-\infty}^0\!\!\xi_2f_o\omega^{[0,\theta]}\mathrm{d}\xi_2\right)^2 \\
	&\leq&  \sqrt{\frac{2\pi}{\theta}}\int_{-\infty}^0\!\!\xi_2\left(f_e^2 + f_o^2
	\right)\omega^{[0,\theta]}\mathrm{d}\xi_2 - \frac{1}{S(0,0)}\frac{2\pi}{\theta}\left(\int_{-\infty}^0\!\!
	-\xi_2f_o^2 \omega^{[0,\theta]} \mathrm{d}\xi_2\right)\left(\int_{-\infty}^0\!\!
	-\xi_2\omega^{[0,\theta]} \mathrm{d}\xi_2\right)\\
	&=& \sqrt{\frac{2\pi}{\theta}}\int_{-\infty}^0\!\!
	\xi_2f_e^2\omega^{[0,\theta]} \mathrm{d}\xi_2 \leq 0.
\end{eqnarray*}
Note that here 
	$S(0,0)=\displaystyle\sqrt{\frac{2\pi}{\theta}}\displaystyle{\int_{-\infty}^0}
	\!\!\xi_2\omega^{[0,\theta]} \mathrm{d}\xi_2=-1$.
	If $x^T \boldsymbol{T}^b x =0$, there must have $f_e=0$ and $f_o$ a constant function, which
	means $x= \boldsymbol{0}\in \mathbb{R}^{m_e+1}$ because $f_o$ is at least a polynomial of degree 2 if
	it's not zero. 
	Thus $\boldsymbol{T}^b$ is negative definite. So is $\+T$ by definition.
\hfill
\end{proof}

%\subsection{Well-posedness.}
\begin{defn}
	We call the ODEs \eqref{eq:bartheta}\eqref{eq:mat-temp} with the
	boundary conditions \eqref{eq:bc1}\eqref{eq:nwbc-2} and a given
	constant $\bar{q}_2$ the $M$-th order reduced moment system for
	the temperature jump problem.
\end{defn}

Note that if $\+M$ has zero eigenvalues, i.e, $M$ is even, the
matrix $\+E$ in \eqref{eq:nwbc-2} would make the analysis more complicated.
And for our purpose, we just need to choose odd $M$ to obtain a series
of temperature solutions for the temperature jump problem. Thus, we have:
\begin{theorem}\label{lemma:temwp}
  For any given constant $\bar{q}_{2}$, accommodation coefficient
  $\chi\in(0,1]$ and odd moment order $M\geq3,M\in\bbN$,
  \eqref{eq:nwbc-2} has a unique solution of $\bar{\theta}(0)$ and
  $\hat{v}_+(0).$
\end{theorem}
\begin{proof}
  By the orthogonal diagonalization \eqref{eq:Rexp},
  $\+M_0\+R _{\text{odd}} = \+R _{\text{even}} \+\Lambda_+$ and
  $\+R _{\text{even}} ^T \+R _{\text{even}} = \frac{1}{2}\+I_{m_o}$.
  Note the diagonal matrix $\+L_1= \mathrm{diag}(a_i)$ is invertible,
  we can write \eqref{eq:nwbc-2} as
  \begin{eqnarray}\label{eq:0429}
    \+E\+K(\chi) \begin{bmatrix}
      1 & \\ & \+R _{\text{even}}
    \end{bmatrix}\begin{bmatrix}
      \bar{\theta}(0)-\bar{\theta}^W \\ \hat{v}_+(0)
    \end{bmatrix} = \bar{q}_2 c_r,\quad
    \+K(\chi):=b(\chi) \boldsymbol{T} - 2\begin{bmatrix}
      0 & \\ & \+R _{\text{even}}\+\Lambda_+\+R
      _{\text{even}} ^T\end{bmatrix},
  \end{eqnarray}
  where $b(\chi)=\frac{2\chi}{(2-\chi)\sqrt{2\pi}} >0$ when
  $\chi\in(0,1]$, we immediately know that $\+K(\chi)$ is negative
  symmetric definite by Lemma \ref{lem:02}.  When $M$ is odd, we
  further have $m_o=m_e=M-2$, $\+E=\+I_{m_o+1}$ is just the identity
  matrix and $\+R _{\text{even}}$ is a square matrix. So
  $\+R _{\text{even}}$ is invertible and \eqref{eq:0429} immediately
  implies that the coefficient matrix is non-singular.  \hfill
\end{proof}
\begin{remark} The well-posedness of the linear kinetic equations is
  widely studied.  We note that for the non-stationary problem, when
  $M$ is odd, the reduced system with the given boundary conditions is
  symmetric hyperbolic with dissipative boundary conditions.  Many
  classical results such as \cite{1960Local} have studied the
  well-posedness of this type of linear problems.
\end{remark}
\begin{remark}\label{rem:odd}
  When $M$ is even, the zero eigenvalues of $\+M$ will make the case
  more complicated.  For two reasons, we think it unnecessary to
  consider this case. One is mentioned previously: the odd $M$ can
  already give a series of solutions. The other is inspired from
  \cite{egg2012}, which shows that when $M$ is even, alternative
  spaces should be used to ensure the stability, i.e. multiplying even
  polynomials $p_{\+\alpha}$ when imposing the boundary conditions.
\end{remark}
\begin{remark}
  Nevertheless, numerically we verify that \eqref{eq:nwbc-2} has a
  unique solution when $M$ is even, varying from 4 to 4000.  In fact
  the coefficient matrix can write as a form $b(\chi)\+S_1-\+S_2$
  where $\+S_1,\+S_2$ is the constant matrix.  We can calculate a
  generalized eigenvalue problem to get $b(\chi)$ such that
  $|b(\chi)\+S_1-\+S_2|=0$ in numeric.
\end{remark}

%\subsection{Solution Profile.}
Now for arbitrary odd moment order $M\geq 3$, we have the formal
analytical solutions of the temperature profile:
\begin{eqnarray}\label{eq:gsol}
  \bar{\theta}(\bar{y})
&= & -\frac{2}{5}\frac{\Pr}{\mathit{Kn}}\bar{q}_2\bar{y} + c_0 - \frac{4}{5}
       \begin{bmatrix}
         \frac{\sqrt{3}}{3},\frac{\sqrt{6}}{2},\frac{\sqrt{2}}{2}
       \end{bmatrix} \boldsymbol{R} _{\text{even}}
       [1:3,:]\exp\left(-\frac{1}{\mathit{Kn}}
         \boldsymbol{\Lambda}_+^{-1}\bar{y}\right)
       \hat{v}_+(0) \notag \\
  &:= & -\frac{2}{5}\frac{\Pr}{\mathit{Kn}}\bar{q}_2\bar{y} + c_0 +
        \sum_{i=1}^{m_o} \tilde{r}_i
        \exp\left(-\frac{1}{\mathit{Kn}}\lambda_i^{-1}\bar{y}\right),
\end{eqnarray}
where $\bar{q}_2$ is the given constant.
$\boldsymbol{R} _{\text{even}} [1:3,:]$ is the first three rows of
$\boldsymbol{R} _{\text{even}}$, $c_0,\ \tilde{r}_i$ are some
constants calculated from the previous process.

Qualitatively, we can see that the temperature profile are
superpositions of Knudsen layers of various widths, which is shown
similarly in \cite{2008Linear, Lijun2017}. When $\bar{y}$ goes to
infinity, $\bar{\theta}$ will asymptotically approach a linear
function, which is exactly the classical Fourier's Law. While as
$\bar{y}$ goes to zero, the phenomenon of temperature jump will occur.
\begin{cor}\label{re:q1}
  From \eqref{eq:nwbc-2} we can see that $\hat{v}_+(0)$ and
  $\bar{\theta}(0)-\bar{\theta}^W$ are depend linearly on $\bar{q}_2$,
  where the coefficients depend on $\chi$, the moment order $M$ but
  not on $\bar{q}_2$. Further we can see from \eqref{eq:gsol} that
  $\bar{\theta}(\bar{y}) -\bar{\theta}^W$ is depend linearly on
  $\bar{q}_2$.
\end{cor}

{\bf Case $M\textbf{=3}.\quad $} As an illustrative example, we show
the case $M=3$ in some detail.  When $M=3$, $m_o=m_e=1$,
$\hat{f}=(\bar{t}_0,\bar{t}_1-0.2\bar{q}_2)^T$, and the eigenvalue
decomposition gives
\begin{eqnarray*}
  \boldsymbol{M} = \begin{bmatrix} 0 & \frac{3}{\sqrt{5}}\\
	  \frac{3}{\sqrt{5}} & 0
\end{bmatrix},\quad
\boldsymbol{\Lambda} = \begin{bmatrix}
\frac{3}{\sqrt{5}} & 0 \\ 0 & -\frac{3}{\sqrt{5}}
\end{bmatrix},\quad
\boldsymbol{R} = \begin{bmatrix}
-\frac{\sqrt{2}}{2} & -\frac{\sqrt{2}}{2} \\
	-\frac{\sqrt{2}}{2} & \frac{\sqrt{2}}{2}
\end{bmatrix}.
\end{eqnarray*}
The boundary condition \eqref{eq:bc1} is
$\bar{t}_0(0)=\sqrt{5}(\bar{t}_1(0)-0.2\bar{q}_2)$ and
\eqref{eq:nwbc-2} becomes
\begin{eqnarray*}
\begin{bmatrix} 1\\0.8\end{bmatrix}
\bar{q}_2 + \begin{bmatrix} 0 & \\ & 9 \end{bmatrix} 
\begin{bmatrix} \bar{\theta}(0)-\bar{\theta}^W \\ 
	\bar{t}_1(0)-0.2\bar{q}_2
\end{bmatrix} = b(\chi)\begin{bmatrix} -2 & -1\\ -1 & -5
\end{bmatrix}
\begin{bmatrix} \bar{\theta}(0)-\bar{\theta}^W \\
	\bar{t}_0(0)
\end{bmatrix}.
\end{eqnarray*}
So after some calculations we have \bea
\bar{\theta}(\bar{y})=-\frac{2}{5\Kn}\Pr\bar{q}_2\bar{y}+c_0-0.8\bar{t}_0(0)
\exp\left(-\frac{\sqrt{5}}{3\Kn}\bar{y}\right), \eea where $\bar{q}_2$
is a given constant, \beaa \bar{t}_0(0)
=-\frac{\bar{q}_2}{\sqrt{5}(6+3\sqrt{5}b(\chi))},\quad c_0 =
-\frac{\bar{q}_2}{2b(\chi)} + \bar{\theta}^W + 0.3\bar{t}_0(0).  \eeaa

For general $M$, we need to determine the coefficients in the
solutions numerically. Since $\+M_0$ is lower-triangular with
bandwidth three, the eigenvalue decomposition will somehow be standard.
And when $M$ is odd, the linear system \eqref{eq:0429} can be
symmetric definite, which benefits the linear solver too.

%%% Local Variables:
%%% mode: latex
%%% TeX-master: "article"
%%% End:

\section{The Kramers' Problem}
\par Before the quantitative study of the temperature profile, 
we will briefly represent the results of the Kramers' problem
for further reference convenience. All the proof is analogous to the 
temperature jump case and will be skipped.
Our previous work \cite{Lijun2017} has studied the velocity profile 
in the Kramers' problem for LHME with the BGK collision model but not 
explicitly considering this even-odd parity form as well as the 
Shakhov collision model.

%\subsection{The Basic Equations.} 
In this section, we will use the script $k$ to represent
Kramers' problem. As mentioned in Section \ref{sec:25}, we set 
$\boldsymbol{\alpha}= \boldsymbol{e}_1+ i \boldsymbol{e}_2,\
0\leq i \leq M-1$ in the $M$-th order $(M>3)$ LHME \eqref{eq:simHME}
to decouple the following $M$ equations involving $\bar{u}_1$:
\begin{eqnarray}\label{eq:u}
\od{\bar{\sigma}_{12}}{\bar{y}} = 0,\quad 
	\od{\bar{u}_1}{\bar{y}}&=& -\frac{1}{\Kn} \bar{\sigma}_{12}- 2\od{\vf{2}}{\bar{y}}, \\
	\boldsymbol{M}_k \od{\hat{w}_k}{\bar{y}} &:=& \begin{bmatrix} \+0 & \+M_0^k \\ (\+M_0^k)^T & \+0
	\end{bmatrix} \od{\hat{w}_k}{\bar{y}} = -\frac{1}{\Kn}\hat{w}_k\label{eq:kramer},
\end{eqnarray}
where $\hat{w}_k=\+L_k\hat{f}_k:=(\hat{w}^k _{\text{even}}, \hat{w}^k _{\text{odd}})^T$,
$\+L_k= \mathrm{diag}(\+L^k_{1},\+L^k_{2}),\ \hat{f}_k=(\hat{f}^k _{\text{even}},\hat{f}^k _{\text{odd}} )^T$.
The index $m_e^k = \lfloor \frac{M-1}{2} \rfloor,\quad m_o^k = \lfloor \frac{M-2}{2} \rfloor.$
Similarly $\hat{f}^k _{\text{even}} := (\vf{2},\vf{4},...,\vfb{2m_e^k})^T\in \mathbb{R}^{m_e^k}$ collects
the even subscripts, and 
$\hat{f}^k _{\text{odd}} := (\vf{3},\ \vf{5},...,\vfb{2m_o^k+1})^T\in \mathbb{R}^{m_o^k}$ collects the odd. 
Here 
$\boldsymbol{M}_0^k=(m_{ij}^{k,0})\in \mathbb{R}^{m_e^k\times m_o^k}$, 
$\boldsymbol{L}^k_1= \mathrm{diag}(a_i^k)_{i=1}^{m_e^k}$,
$\boldsymbol{L}^k_2= \mathrm{diag}(b_i^k)_{i=1}^{m_o^k}$ have the entries:
\begin{eqnarray} \label{eq:m12}
	m_{ij}^{k,0} = \frac{1}{a_ib_j}\langle \phi_i^k, \xi_2\varphi_j^k \rangle_{\bar{\omega}},\ 
	a_i^k = \sqrt{\left(1- \frac{1-\Pr}{5}\delta_{i,1}\right) \langle \phi_i^k, \phi_i^k\rangle_{\bar{\omega}}} ,\ 
	b_i^k = \sqrt{\langle \varphi_i^k, \varphi_i^k\rangle_{\bar{\omega}}},
\end{eqnarray}
where $\phi_i^k:= \overline{\mathrm{He}}_{\boldsymbol{e}_1 + 2i \boldsymbol{e}_2}$, 
$\varphi_j^k:= \overline{\mathrm{He}}_{\boldsymbol{e}_1 + (2j+1) \boldsymbol{e}_2}$.
\begin{lemma} $\+M_k$ has $m_o^k$ positive,
$m_o^k$ negative and $m_e^k-m_o^k$ zero eigenvalues. 
\end{lemma}
\begin{cor}
	There exists a real orthogonal diagonalization 
	$\boldsymbol{M}_k \boldsymbol{R}_k = \boldsymbol{R}_k \boldsymbol{\Lambda}_k$ where
	\begin{eqnarray}
\label{eq:evdofv}
\boldsymbol{R}_k:=
\begin{bmatrix}
\boldsymbol{R}^k _{\text{even}} & \boldsymbol{R}^k_0 & \boldsymbol{R}^k _{\text{even}} \\
\boldsymbol{R}^k _{\text{odd}} & \boldsymbol{0} & - \boldsymbol{R}^k _{\text{odd}}
\end{bmatrix},\quad
	\boldsymbol{\Lambda}_k := \begin{bmatrix} \boldsymbol{\Lambda}_+^k & & \\
	& \boldsymbol{0}_{m_e^k-m_o^k} & \\
		& & - \boldsymbol{\Lambda}_+^k
	\end{bmatrix}.
\end{eqnarray}
Here $\boldsymbol{R}_k$ is orthogonal, 
	$\boldsymbol{R}^k _{\text{even}}\in \mathbb{R}^{m_e^k\times m_o^k},\ \boldsymbol{R}^k _{\text{odd}}
	\in \boldsymbol{R}^{m_o^k\times m_o^k},\ 
	\boldsymbol{R}_0^k\in \boldsymbol{R}^{m_e^k\times(m_e^k-m_o^k)}$, and
	$\boldsymbol{\Lambda}_+^k:= \mathrm{diag}(\lambda_{k,i})\in \mathbb{R}^{m_o^k\times m_o^k}$ 
	with $\lambda_{k,1}\geq\lambda_{k,2}\geq\cdots\geq\lambda_{k,m_o^k}>0$.
\end{cor}

%\subsection{Boundary Conditions.} 
Similarly define $\hat{v}_k= \boldsymbol{R}_k^{-1} \hat{w}_k
=(\hat{v}_{k,+},\hat{v}_{k,0},\hat{v}_{k,-})^T$, 
then \eqref{eq:kramer} will turn to $\displaystyle\+\Lambda_k\od{\hat{v}_k}{\bar{y}}=-\frac{1}{\Kn}\hat{v}_k$.
The boundedness and consistency asks $m_e^k$ boundary conditions
	\begin{eqnarray}\label{eq:vbc1}
		\hat{v}_{k,0}(0) = \+0,\quad \hat{v}_{k,-}(0)=\+0.
	\end{eqnarray}
Setting $\boldsymbol{\alpha} = \boldsymbol{e}_1 + i \boldsymbol{e}_2,\ 0\leq i\leq M-2,\ i$ even, 
in \eqref{eq:lbc-0317} to get wall boundary conditions:
	\begin{eqnarray}
	\label{eq:vbc2}
	\bar{\sigma}_{12} r_k + \+E_k\begin{bmatrix} 0 & \\
	& \boldsymbol{M}_0^k\end{bmatrix}\begin{bmatrix}
	\bar{u}_1(0)-\bar{u}_1^W \\ \hat{w}^k _{\text{odd}} (0) \end{bmatrix} = b(\chi) \+E_k\tilde{\boldsymbol{S}_k}
		\begin{bmatrix}
		\bar{u}_1(0)-\bar{u}_1^W \\ \hat{w}^k _{\text{even}} (0) \end{bmatrix},
	\end{eqnarray}
	where $r_k=(1,2/a_1^k,0,...,0)^T,$
	$\tilde{\+S}_k:= \mathrm{diag}(1,\+L_1^k)^{-1}\+S_k
	\mathrm{diag}(1,\+L_1^k)^{-1}$, $\+E_k=[\+I_{m_o^k+1},\+0]$ is a $(m_o^k+1)\times (m_e^k+1)$ matrix.
	Here
	$\boldsymbol{S}_k=(s_{ij}^k)\in \mathbb{R}^{(m_e^k+1)\times (m_e^k+1)}$ has the entries
	\begin{eqnarray*}
	s_{ij}^k = S(2i-2,\ 2j-2),\ i,j\geq 1.
	\end{eqnarray*}
	Thus substituting $\hat{w}_k=\+R_k\hat{v}_k$ and \eqref{eq:vbc1} into \eqref{eq:vbc2}, we can determine
	$\hat{v}_{k,+}(0)$ and $\bar{u}_1(0)$.
 	\begin{lemma} $\+S_k$ is negative symmetric definite.
	\end{lemma}
	\begin{theorem} For any given constant $\bar{\sigma}_{12}$,
		accommodation coefficient $\chi\in(0,1]$ and even moment order $M\geq4,M\in\bbN$,
		\eqref{eq:vbc2} has a unique solution of $\bar{u}_1(0)$ and $\hat{v}_{k,+}(0).$
	\end{theorem}
	
	%\subsection{Solution Profile.}	
Finally the velocity solution has the form
	\vspace{-3mm}
\begin{eqnarray*}
\bar{u}_1(\bar{y}) = -\frac{1}{\Kn}\bar{\sigma}_{12}\bar{y} + c_0^k -\frac{2}{a_1^k}\+R^k _{\text{even}} [1,:]
\exp\left(-\frac{1}{\Kn}(\+\Lambda_+^k)^{-1}\bar{y}\right)\hat{v}_{k,+}(0),
\end{eqnarray*}
where all constants can be determined by eigenvalue decomposition
and linear solvers.

\section{Numerical Validation}

In this section, we will represent some numerical results in the
temperature jump problem. As in the kinetic theory, we consider the
normalized temperature combined by three parts
\begin{eqnarray*}
\tilde{\theta}(\bar{y}) = \bar{y} + \zeta - \theta_d(\bar{y}),
\end{eqnarray*}
where $\bar{y}$ is the linear part, $\theta_d(\bar{y})$ is the 
temperature defect satisfying
$\displaystyle\lim_{\bar{y}\rightarrow\infty}\theta_d(\bar{y}) = 0$
and $\zeta$ is the temperature jump coefficient. In our model, 
\bea
\label{eq:ntemp}
\zeta = -\frac{5\Kn}{2\Pr\bar{q}_2}c_0,\quad \theta_d =
-\frac{2\Kn}{\Pr\bar{q}_2}(\bar{t}_0+
(1-\delta_{M,3})(6\bar{t}_2+\bar{s}_2)).
\eea
As shown in Corollary \ref{re:q1},
$c_0-\bar{\theta}^W,\bar{t}_0,\bar{t}_2,\bar{s}_2$ 
are linear dependent on $\bar{q}_2$, so we may as well set
$\bar{\theta}^W=0$ and $\bar{q}_2=1$. Since the reasons in
Remark \ref{rem:odd}, we just consider the case when $M$ is odd.

\subsection{Temperature Jump Coefficient $\zeta$.}
We compare the temperature jump coefficient when $M=2k+1,\ 1\leq k\leq 6$
with the results solved by discrete-ordinates methods of linearized
Boltzmann-BGK model \cite{Siewert2000} in Table.\ref{tab:01}. 
The parameters are chosen to be consistent with \cite{Siewert2000},
i.e. $\Kn=\frac{\sqrt{2}}{2},\ \Pr=1.$
\begin{table}[!htb] 
\centering 
\caption{The temperature jump coefficient compared with Barichello 
and Siewert's results \cite{Siewert2000}}\label{tab:01}
\begin{tabular}{cccccccc} 
\toprule 
$\chi$ & Siewert's & $M=3$ & $M=5$ & $M=7$ & $M=9$ & $M=11$& $M=13$\\ 
\midrule 
0.1 & 21.45012 & 21.086 & 21.357 & 21.396 & 21.412 & 21.421 & 21.426\\
0.3 & 6.630514 & 6.3116& 6.5542 & 6.5870 & 6.6003 & 6.6074 & 6.6118\\
0.5 & 3.629125 & 3.3538& 3.5680 & 3.5951 & 3.6057 & 3.6114 & 3.6149\\
0.6 & 2.867615 & 2.6134& 2.8135 & 2.8378 & 2.8473 & 2.8522 & 2.8553\\
0.7 & 2.317534 & 2.0840& 2.2698 & 2.2916 & 2.3000 & 2.3043 & 2.3070\\
0.9 & 1.570264 & 1.3768& 1.5342 & 1.5513 & 1.5576 & 1.5608 & 1.5628\\
1.0 & 1.302716 & 1.1287& 1.2718 & 1.2867 & 1.2921 & 1.2949 & 1.2965\\
\bottomrule 
\end{tabular}
\end{table}

As can be seen, when $\chi$ becomes smaller the temperature jump
coefficient will go larger. And for the given $\chi$, the LHME solutions 
seem to agree with the reference solutions with not too many moments. 
In fact when $M=13$ the relative error between the LHME solution 
and the reference solution is less than $1\%$ in most cases. There is
also a convergence trend when $M$ grows. In fact if $\zeta_k$ is the
LHME solution when $M=2^{k}+1$, we can define the numerical 
convergence order as
\beaa
\beta_k = -\log_2\left(\frac{\zeta_{k+2}-\zeta_{k+1}}
{\zeta_{k+1}-\zeta_k}\right).
\eeaa
Table.\ref{tab:add1} shows $\beta_k$ when $\chi$ is different and $k=6,7,8$.
The results imply about one order convergence when $M\rightarrow\infty$
and the accuracy of the linear solver may impact on $\beta_k$
when $M$ is large.
%\vspace{-6mm}
\begin{table}[!htb] 
\centering 
\caption{The numerical convergence order of the temperature jump coefficient}\label{tab:add1}
\begin{tabular}{cccccccc} 
\toprule 
	$\chi$ & 0.1 & 0.3 & 0.5 & 0.6 & 0.7 & 0.9 & 1.0\\ 
\midrule 
$k=6$ & 0.984 & 0.995 & 1.006 & 1.012 & 1.018 & 1.029 & 1.036\\
$k=7$ & 0.976 & 0.985 & 0.993 & 0.998 & 1.003 & 1.012 & 1.017\\
$k=8$ & 0.974 & 0.981 & 0.988 & 0.991 & 0.995 & 1.002 & 1.006\\
\bottomrule 
\end{tabular}
\end{table}

\begin{remark}
	For the linearized moment system, we think its capacity to
	describe the Knudsen layer mainly lies in the approximation of
	basis function spaces, i.e. similarly as in Galerkin spectral methods,
	but is rarely dependent on the hyperbolic regularization. For the HME,
	we can see from Remark \ref{rem:Grad} that the hyperbolic
	regularization does not affect the linearized moment system. Beyond,
	it may be also true for 13 or 26 moment methods.
\end{remark}

Since $\Pr$ only occurs in \eqref{eq:bartheta} and does not affect the
other equations or boundary conditions, if we seem $\zeta$ as a
function of $\Pr$, immediately we have
\bea\label{eq:Prlin}
\zeta(\Pr) = {\Pr}^{-1}\zeta(1).
\eea
So we can just consider the BGK model when studying the jump
coefficient. We note this relation 
\eqref{eq:Prlin} is also shown in \cite{2019Hat} when studying 
the Shakhov model.

Fig. \ref{fig:add1} shows the value $b(\chi)\zeta$ 
when $\chi$ is different and $M$ is fixed.
\vspace{-3mm}
\begin{figure}[htb!]
\centering
\newcommand\addsubfigure[2]{{%
    \begin{overpic}[width=#2 \textwidth]{#1}
    \end{overpic}}
}
{%
\addsubfigure{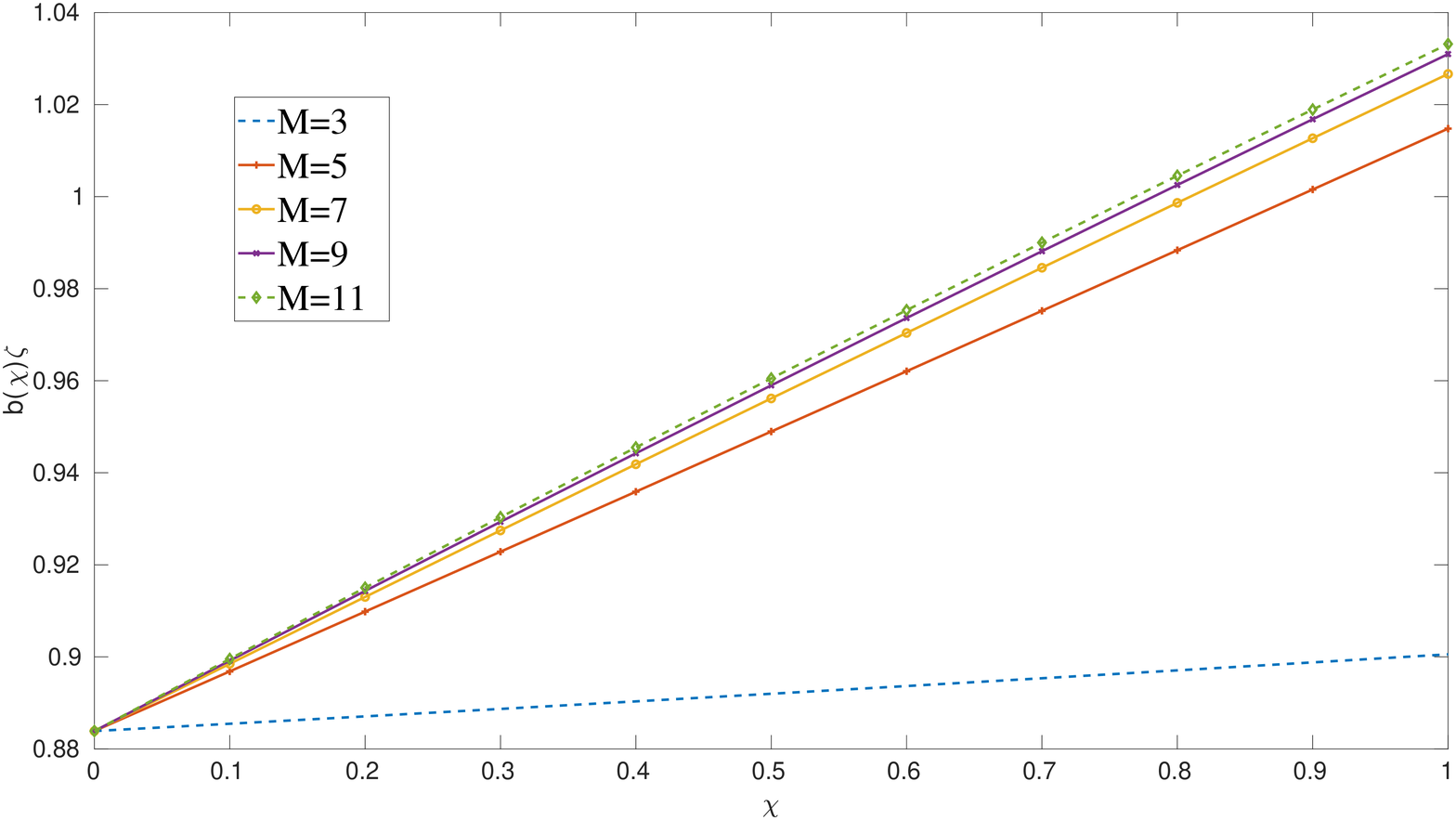}{0.92}
}
\caption{\label{fig:add1}
The value of $b(\chi)\zeta$ for the LHME solutions.
}
\end{figure}%}

In fact we can formally show the convergence results 
when $\chi\rightarrow0$. When 
$\chi\rightarrow 0,\ b(\chi)=\frac{2\chi}{2-\chi}(2\pi)^{-\frac{1}{2}}$
also goes to zero. Note that
$\hat{w} _{\text{odd}}=\+R _{\text{odd}} \hat{v}_+$ and
$\hat{w} _{\text{even}}=\+R _{\text{even}} \hat{v}_+$
should have the same order since $\+R$ remains the same
when $\chi$ varies. So to make both sides of  
\eqref{eq:nwbc-2} the same order, one must 
assume $\hat{w}(0)=O(b\bar{\theta}(0))$.
Thus, the first row of the leading order equations will be
\beaa
-2b\bar{\theta}(0) = \bar{q}_2.
\eeaa
After the normalization and note that
$c_0(0)=\bar{\theta}(0)+o(\bar{\theta}(0))$, we have
\bea\label{eq:lim1}
\lim_{\chi\rightarrow 0}b\zeta = \frac{5}{8}\sqrt{2}\quad 
\Rightarrow\quad \lim_{\chi\rightarrow 0} \frac{\chi}{2-\chi}\zeta
= \frac{5}{8}\sqrt{\pi}. 
\eea
The limit \eqref{eq:lim1} exactly agrees with the result in
linearized Boltzmann-BGK model as in \cite{1978Temperature}.
\par In a word, the numerical results tell that we may 
only need a moment system with moderate moment order (such as $M=11,\ 13$)
to describe the Knudsen layer in this problem. Since the
1D assumptions, the number of moments
is linearly correlated with the moment order $M$, so this 
scale may be affordable in practice.

\subsection{Temperature Defect $\theta_d(\bar{y})$.}
Fig.\ref{fig:101} presents the profile of the temperature defect
$\theta_d(\bar{y})$ for the LHME when $M=3,7,11,15$ and $\chi=0.1,\ 1.0$.
The reference solution is from the linearized Boltzmann-BGK
model \cite{Siewert2000}. As we can see, the result of $M=3$ is away
from the reference solution but when $M$ becomes larger our results
quickly agree with the reference solution well. When $M\geq 7$, the
gap seems to mainly occur only near the wall, i.e. $\bar{y}$ close to
zero.  
%{\setlength{\belowdisplayskip}{6pt}
\begin{figure}[htb!]
\centering
\newcommand\addsubfigure[2]{{%
    \begin{overpic}[width=#2 \textwidth]{#1}
    \end{overpic}}
}
{%
\addsubfigure{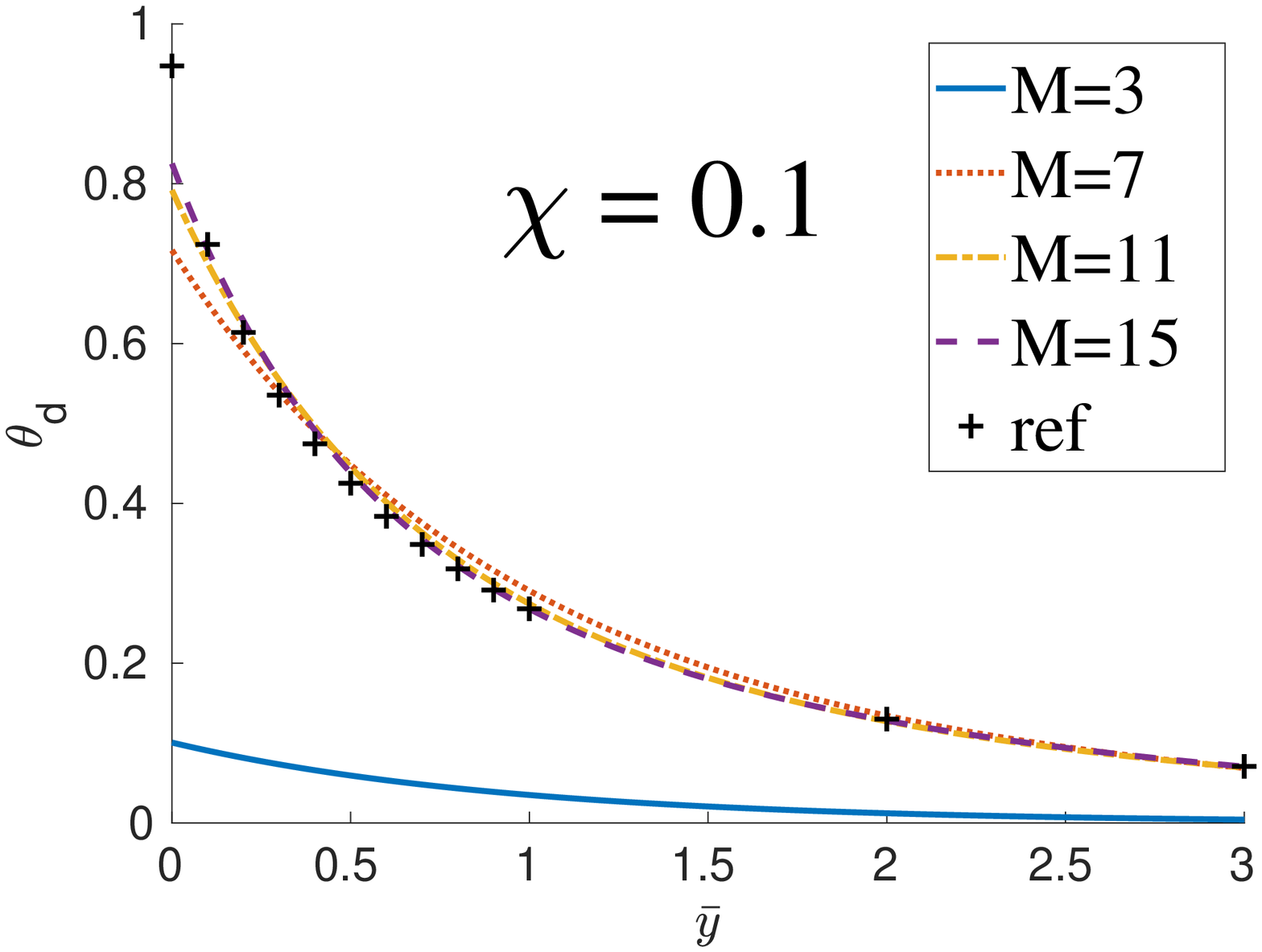}{0.49}
\addsubfigure{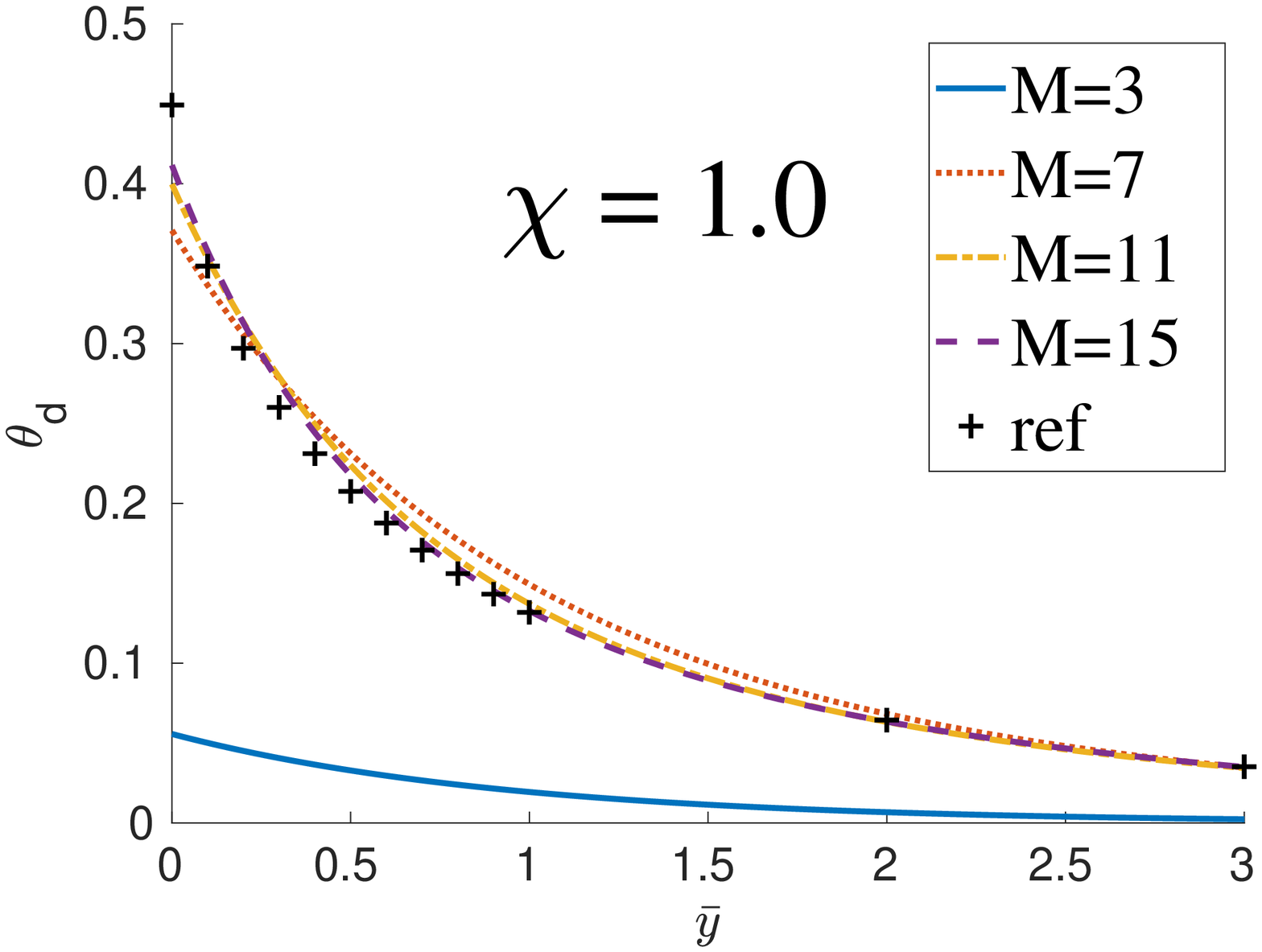}{0.49}
}
\caption{\label{fig:101}
Profile of the temperature defect $\theta_d(\bar{y})$ of the LHME for
different $M$ with $\chi=0.1,\ 1.0$.
}
\end{figure}%}

In fact the analytical expressions of $\theta_d(\bar{y})$ are
available in our model, by \eqref{eq:gsol},
\bea\label{eq:liscoef}
\theta_d(\bar{y}) &=& -\frac{2\Kn}{\Pr\bar{q}_2}
\begin{bmatrix}
\frac{\sqrt{3}}{3},\frac{\sqrt{6}}{2},\frac{\sqrt{2}}{2}
\end{bmatrix} \boldsymbol{R} _{\text{even}} [1:3,:]
\exp\left(-\frac{1}{\mathit{Kn}}
	\boldsymbol{\Lambda}_+^{-1}\bar{y}\right)
	\hat{v}_+(0) \notag\\
	&:=& -\frac{2\Kn}{\Pr}
	\sum_{i=1}^{m_o} \tilde{c}_i \exp\left(-\frac{1}{\mathit{Kn}}
	\lambda_i^{-1}\bar{y}\right),
\eea
where $\tilde{c}_i$ is only dependent on the moment order $M$ and
accommodation coefficient $\chi$. The profile of the temperature
defect in Fig. \ref{fig:101} is obtained by solving the constants 
in \eqref{eq:liscoef} then plotting the analytical expressions. 

Again we find that if we want to capture the behavior of the gas near the 
wall, it may be necessary to enlarge the moment order $M$, but 
a modest $M$ such as 11 may be enough considering the balance of
accuracy and efficiency. In other problems, the moment order $M$
should be of concrete analysis.

\subsection{Effective Thermal Conductivity.} 
The Fourier law fails in the Knudsen layer and
we can formally write the Fourier law by the effective thermal conductivity 
$\kappa_{\text{eff}}$:
\bea
q_2 = -\kappa_{\text{eff}}\od{\theta}{y}.
\eea
So if we denote by $\kappa_0$ the original thermal conductivity, we have
\bea
\frac{\kappa_{\text{eff}}}{\kappa_0} = \left(\od{\tilde{\theta}}{\bar{y}}\right)^{-1}
= \left(1-\od{\theta_d}{\bar{y}}\right)^{-1}.
\eea
\begin{figure}[htb!]
\centering
\newcommand\addsubfigure[2]{{%
    \begin{overpic}[width=#2 \textwidth]{#1}
    \end{overpic}}
}
{%
\addsubfigure{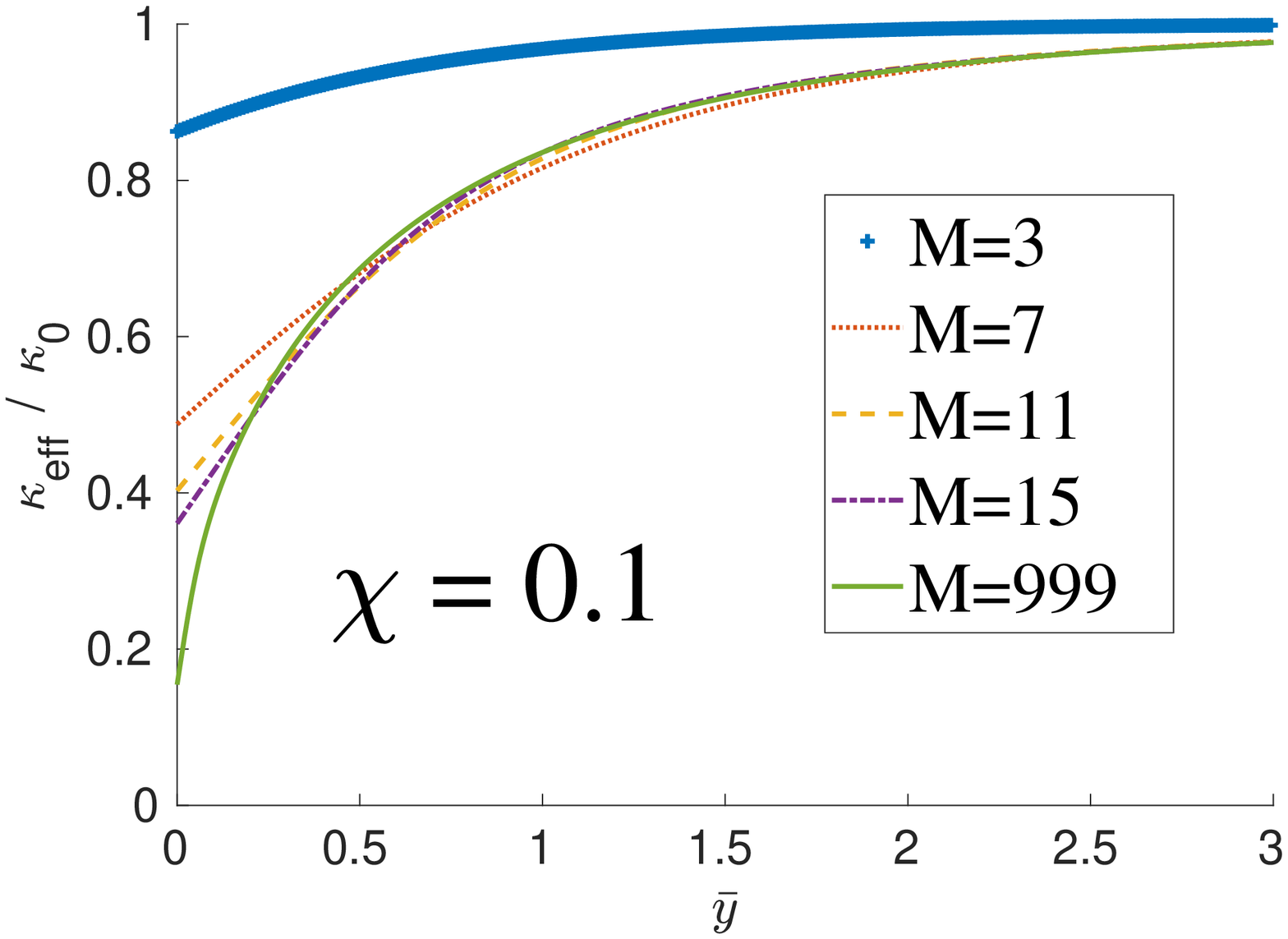}{0.48}%{0.35}
\addsubfigure{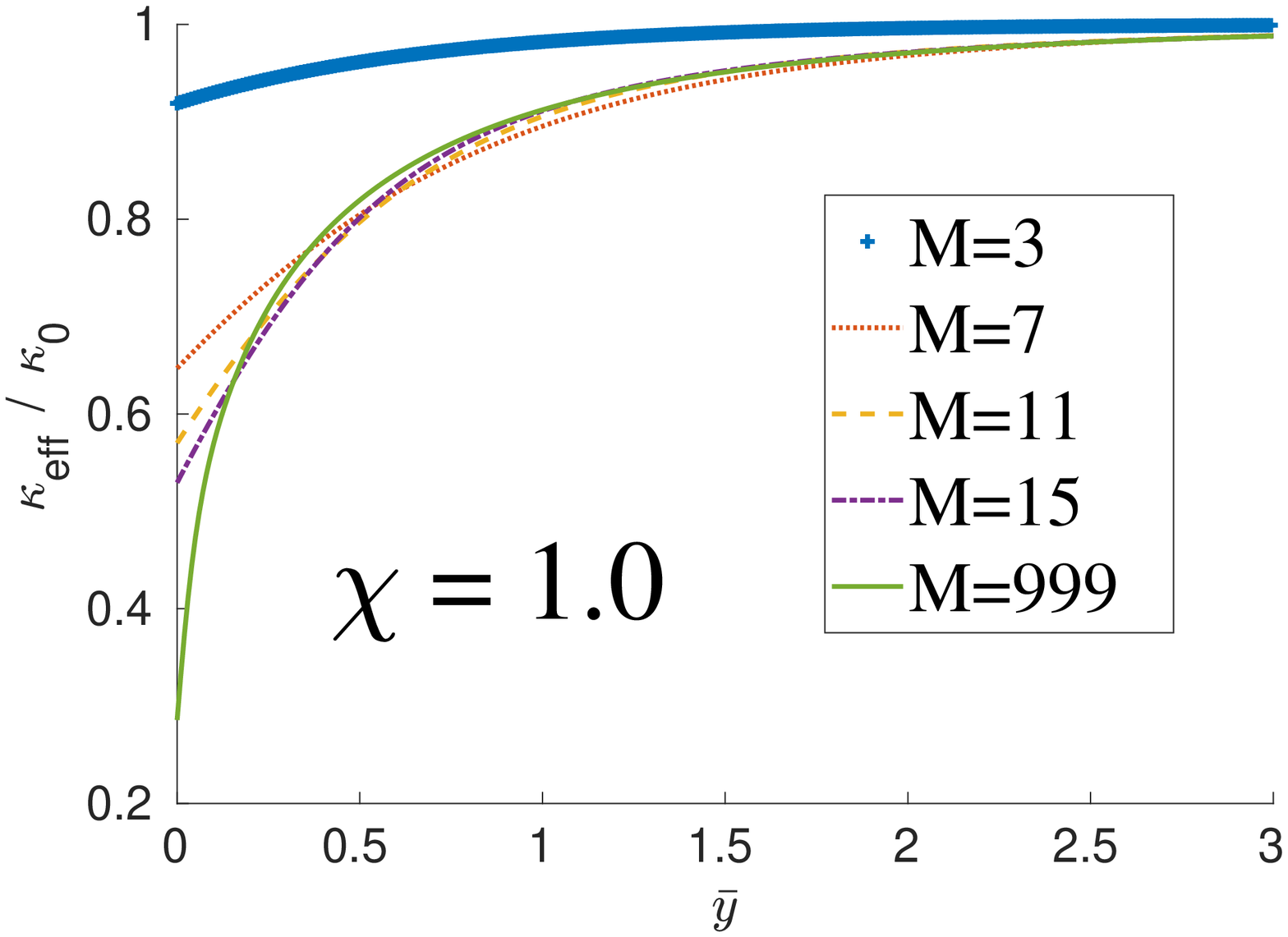}{0.48}%{0.35}
}
\caption{\label{fig:102}
Profile of the effective thermal conductivity in the
Knudsen layer.
}
\end{figure} 

Here we choose $\Pr=\frac{2}{3}$ for the Maxwell molecules and 
study the effective thermal conductivity of different $M$ and
accommodation coefficients numerically in Fig.\ref{fig:102}. We notice
that \cite{Gu2014} obtained a similar form of
$\kappa_{\text{eff}}/\kappa_0$ with exponential terms by the R26
moment system and compare our results with it. We can find that the
LHME captures all the qualitative trends of $\kappa_{\text{eff}}$
mentioned in \cite{Gu2014} in the Knudsen layer, such as
$\kappa_{\text{eff}}$ will reduce as $\bar{y}\rightarrow 0$ or
$\chi\rightarrow 0$.

These analytical expressions may be used to correct the boundary
conditions of the NSF equations. But since $\kappa_{\text{eff}}$
relies on the flow conditions \cite{Gu2014}, we must be very careful
in the application. This may be the future study and we don't plan 
to deal with it in this paper.

\section{Conclusions}
\par We have derived an approximate analytical solution for the Knudsen 
layer using arbitrary high order LHME. A class of well-posed boundary
conditions for the LHME under all accommodation coefficients has been
imposed. And the formal analytical solutions with some constants
determined by numerical solvers have been presented. In the
temperature jump problem, we have compared the
temperature defect, temperature jump coefficient and effective thermal
conductivity of our model with the existing models. It's shown that
the LHME with a few moments can capture the thermal Knudsen layer
well. Although we restricted us mainly in the temperature jump problem
(and Kramers' problem), it is straightforward to extend the method to
other boundary layer problems as well as other collision models.

\section*{Acknowledgement}
\par This work is financially supported by the National Key R\&D
Program of China, Project Number 2020YFA0712000. We thank Dr.Jun Li
for her enthusiastic discussions with us.

%%% Local Variables:
%%% mode: latex
%%% TeX-master: "article"
%%% End:

%\begin{appendix}
% 参考文献格式。
\bibliographystyle{plain}
% 参考文献。
\bibliography{article}
%\end{appendix}
%\section*{Appendix}
\appendix
\section{Hermite polynomials and the half-space integral}\label{app:A}
\begin{defn} Given $\boldsymbol{u}\in \mathbb{R}^N,\ \theta\in \mathbb{R},\ \theta>0$,
	we define the weight function $\omega ^{[\boldsymbol{u},\theta]}(\boldsymbol{\xi})$, the generalized $N$-D
	Hermite function
$\mathcal{H}_{\boldsymbol{\alpha}}^{[\boldsymbol{u},\theta]}(\boldsymbol{\xi})$ and the Hermite polynomial
$\mathrm{He}_{\boldsymbol{\alpha}}^{[\boldsymbol{u},\theta]}(\boldsymbol{\xi})$ as follows:
\begin{eqnarray}
\omega ^{[\boldsymbol{u},\theta]}(\boldsymbol{\xi}) &=& \frac{1}{(2\pi\theta)^{N/2}}\exp\left(-\frac{|\boldsymbol{\xi}- \boldsymbol{u}|^2}{2\theta}\right),\\
	\mathrm{He}_{\boldsymbol{\alpha}}^{[\boldsymbol{u},\theta]}(\boldsymbol{\xi}) &=& \frac{(-1)^{|\boldsymbol{\alpha}|}}{\omega ^{[\boldsymbol{u},\theta]}(\boldsymbol{\xi})}
	\dfrac{\partial^{|\boldsymbol{\alpha}|} \omega ^{[\boldsymbol{u},\theta]}(\boldsymbol{\xi}) }{\partial \boldsymbol{\xi}^{\boldsymbol{\alpha}}},\\
\mathcal{H}_{\boldsymbol{\alpha}}^{[\boldsymbol{u},\theta]}(\boldsymbol{\xi}) &=& \omega ^{[\boldsymbol{u},\theta]}(\boldsymbol{\xi}) \mathrm{He}_{\boldsymbol{\alpha}}^{[\boldsymbol{u},\theta]}(\boldsymbol{\xi}),
\end{eqnarray}
where 
$\boldsymbol{\alpha}=(\alpha_i)_{i=1}^N\in \mathbb{N}^N$, $|\+\alpha|:=\sum_i \alpha_i$, 
 $\boldsymbol{\xi}\in \mathbb{R}^N$ and $\+\xi^{\boldsymbol{\alpha}}:=\prod_i\xi_i^{\alpha_i}$. 
From the definition we have
\begin{eqnarray}
	\omega ^{[\boldsymbol{u},\theta]}(\boldsymbol{\xi}) = \prod_{i=1}^N \omega^{[u_i,\theta]}(\xi_i),\ 
	\mathrm{He}_{\boldsymbol{\alpha}}^{[\boldsymbol{u},\theta]}(\boldsymbol{\xi}) = \prod_{i=1}^N \mathrm{He}_{\alpha_i}^{[u_i,\theta]}(\xi_i),
\end{eqnarray}
so properties of $N$-D Hermite polynomials will reduce to the 1D case.
\end{defn}
\begin{prop}\label{propA1} When $N=1$,\ i.e. $u,\xi\in \mathbb{R},\alpha\in \mathbb{N}$, we have
	(proof in \cite{Fan_new} or anywhere)
\begin{itemize}
	\item Recursion relation: $(\xi-u) \mathrm{He}_{\alpha+1}^{[u,\theta]}(\xi) = 
		(\alpha+1)\mathrm{He}_{\alpha}^{[u,\theta]}(\xi)+\theta \mathrm{He}_{\alpha+2}^{[u,\theta]}(\xi)$.
	\item Differential relation I: $\pdd{\xi}\mathrm{He}_{\alpha+1}^{[u,\theta]}(\xi) 
		= \dfrac{\alpha+1}{\theta} \mathrm{He}_{\alpha}^{[u,\theta]}(\xi)$.
	\item Differential relation II: $\pdd{\xi}\mathcal{H}_{\alpha}^{[u,\theta]}(\xi)
		= - \mathcal{H}_{\alpha+1}^{[u,\theta]}(\xi)$.
	\item Orthogonal relation: $\displaystyle\int_{\mathbb{R}}\!\!\mathrm{He}_{\alpha}^{[u,\theta]}(\xi)
		\mathcal{H}_{\beta}^{[u,\theta]}(\xi)\mathrm{d}\xi = 
		\langle \mathrm{He}_{\alpha}^{[u,\theta]},\mathrm{He}_{\beta}^{[u,\theta]} \rangle
		_{\omega^{[u,\theta]}}
		= \alpha!\theta^{-\alpha}\delta_{\alpha,\beta}.$
\end{itemize}
\end{prop}

\begin{prop} \label{prop:1} The half-space integral $S(\alpha_2,\beta_2)$ defined as \eqref{eq:defS} 
	is independent of $\theta$. 
	And \begin{eqnarray*} S(\alpha_2,\alpha_2+1)=\frac{\sqrt{2\pi}}{2}(\alpha_2+1)!;\quad 
	S(\alpha_2,\alpha_2-1)=\frac{\sqrt{2\pi}}{2}\alpha_2!\ (\alpha_2>0). \end{eqnarray*} Otherwise when
	$\beta_2\neq \alpha_2-1,\alpha_2+1$, we have
\begin{eqnarray*}
S(\alpha_2,\beta_2) = \frac{\alpha_2+\beta_2+1}{(\alpha_2-\beta_2)^2-1}
z_{\alpha_2}z_{\beta_2},
\end{eqnarray*}
where $z_{0}=1,\ z_1=0,\ z_{n+1}=-nz_{n-1},\ n\geq 1$. 
\end{prop}
\begin{cor}  Since $z_n=0$ when $n$ is odd, from Property \ref{prop:1}, 
	we have $S(\alpha_2,\beta_2)=0$ if $\alpha_2$ is even, $\beta_2$
	is odd and $|\beta_2-\alpha_2|\neq1$.
\end{cor}
\begin{proof}{\bf (Proof of Proposition \ref{prop:1}.)} First for $\alpha,\beta\in \mathbb{N}$, we denote
by \begin{eqnarray*}
I(\alpha,\beta) = \sqrt{2\pi}\int_{-\infty}^0\!\!\theta^{\frac{\alpha+\beta}{2}}
	\mathrm{He}_{\alpha}^{[0,\theta]}(\xi) \mathrm{He}_{\beta}^{[0,\theta]}(\xi)\omega^{[0,\theta]}(\xi)
	\mathrm{d}\xi.
\end{eqnarray*}
So $I(\alpha,\beta)=I(\beta,\alpha)$. Integrate by parts using 
$\mathrm{d}\left(\mathcal{H}_{\beta}^{[0,\theta]}\right)=- \mathcal{H}_{\beta+1}^{[0,\theta]} \mathrm{d}\xi$
or 
$\mathrm{d}\left(\mathcal{H}_{\alpha}^{[0,\theta]}\right)=- \mathcal{H}_{\alpha+1}^{[0,\theta]} \mathrm{d}\xi$,
then we should get the equivalent results by these two ways:
	\begin{eqnarray} \label{eq:AAA1}
	I(\alpha+1,\beta+1) &=& 
	-\sqrt{2\pi}\theta^{\frac{\alpha+\beta+2}{2}}
\mathrm{He}_{\alpha+1}^{[0,\theta]}(0) \mathcal{H}_{\beta}^{[0,\theta]}(0) + (\alpha+1)I(\alpha,\beta) \\
	&=& 
	-\sqrt{2\pi}\theta^{\frac{\alpha+\beta+2}{2}}
\mathrm{He}_{\beta+1}^{[0,\theta]}(0) \mathcal{H}_{\alpha}^{[0,\theta]}(0) + (\beta+1)I(\alpha,\beta).
\end{eqnarray}
Noting that 
$\mathcal{H}_{\alpha}^{[0,\theta]}(0)=(2\pi\theta)^{-\frac{1}{2}}\mathrm{He}_{\alpha}^{[0,\theta]}(0)$.
If we denote by $z_{\alpha}=\theta^{\frac{\alpha}{2}}\mathrm{He}_{\alpha}^{[0,\theta]}(0)$, 
when $\alpha\neq\beta$ we have
\begin{eqnarray}
I(\alpha,\beta) = \frac{1}{\alpha-\beta}
	(z_{\alpha+1}z_{\beta}-z_{\beta+1}z_{\alpha}),	
\end{eqnarray}
	where $z_0=1,\ z_1=0$ and $z_{n+1}=-nz_{n-1}$ by recursion relation in Proposition \ref{propA1}. 
By definition,
\begin{eqnarray}
S(\alpha,\beta) = \beta I(\alpha,\beta-1)+I(\alpha,\beta+1),
\end{eqnarray}
	which turns to $S(\alpha,0)=I(\alpha,1)$ when $\beta=0$.
So when $\beta\neq \alpha+1$ and $\beta\neq \alpha-1$, we have
\bea
\label{eq:AAA3}
S(\alpha,\beta) = -\frac{\beta}{\alpha-\beta+1}z_{\alpha}z_{\beta}-\frac{1}{\alpha-\beta-1}z_{\alpha}
z_{\beta+2} = \frac{\alpha+\beta+1}{(\alpha-\beta)^2-1}z_{\alpha}z_{\beta}.
\eea
	For the special case $\beta=\alpha+1$, we calculate by \eqref{eq:AAA1} to
	get
\beaa
S(\alpha,\alpha+1) = (\alpha+1)I(\alpha,\alpha) = (\alpha+1)!I(0,0) = \frac{\sqrt{2\pi}}{2}(\alpha+1)!.
\eeaa 
Similarly when $\alpha>1$ we have $S(\alpha,\alpha-1)=I(\alpha,\alpha)=\frac{\sqrt{2\pi}}{2}\alpha!$.
\hfill
\end{proof}

\section{Calculation of $\bar{m}_{\boldsymbol{\alpha}}$}\label{app:B}
Assume $\boldsymbol{u}^W=(u^W_1,u^W_2,u^W_3)^T$ and $u^W_2=0$.
Then from $\boldsymbol{u}\cdot \boldsymbol{n}=0$ we have $u_2=0$.
By definition 
\begin{eqnarray}\label{eq:maJ}
	m_{\boldsymbol{\alpha}} &=& \frac{\theta^{|\boldsymbol{\alpha}|}}{\boldsymbol{\alpha} !}
\int_{\mathbb{R}^3}\frac{\rho^W}{\sqrt{2\pi\theta^W}^{3}}
\exp(-\frac{|\boldsymbol{\xi}- \boldsymbol{u}^W|^2}{2\theta^W})
\mathrm{He}_{\boldsymbol{\alpha}}^{[\boldsymbol{u},\theta]}(\boldsymbol{\xi})\mathrm{d} \boldsymbol{\xi} 
	\notag\\ &=& 
\rho^W J_{\alpha_1}(u_1^W-u_1) J_{\alpha_2}(0) J_{\alpha_3}(u_3^W-u_3),
\end{eqnarray}
where the 1D integral $J_{m}(x)$ is defined for $m\in \mathbb{N}$ and $x,u \in \mathbb{R}$ as 
\begin{eqnarray*}
J_{m}(x) := \frac{1}{m!}\theta^{m}
\int_{\mathbb{R}}(2\pi\theta^W)^{-\frac{1}{2}}\exp(-\frac{|\xi-u-x|^2}{2\theta^W})
\mathrm{He}_{m}^{[u,\theta]}(\xi)\mathrm{d}\xi.
\end{eqnarray*}
\begin{prop}
$J_m(x)$ is independent of $u$ and satisfy a recursion relation. 
\end{prop}
\begin{proof}
Use $\mathrm{d}\left(\mathrm{He}_{m+1}^{[u,\theta]}\right)=\frac{m+1}{\theta}\mathrm{He}
_m^{[u,\theta]}\mathrm{d}\xi$ in the integration by parts formula, then we have
\begin{eqnarray}
J_{m}(x) &=& \frac{\theta^{m+1}}{(m+1)!}
\int_{\mathbb{R}}(2\pi\theta^W)^{-\frac{1}{2}}\exp(-\frac{|\xi-u-x|^2}{2\theta^W})
\mathrm{He}_{m+1}^{[u,\theta]}(\xi)\frac{\xi-u-x}{\theta^W}\mathrm{d}\xi \notag\\
&=& \frac{1}{\theta^W}\left(-xJ_{m+1}+\theta J_m(x)+(m+2)J_{m+2}(x)\right),\quad m\geq 0.
\label{eq:app222}
\end{eqnarray} 
Note that $\mathrm{He}_0^{[u,\theta]}(\xi)=1,\ \mathrm{He}_1^{[u,\theta]}(\xi)=(\xi-u)/\theta$, so
if we substitute $m$ by $m-2$ in \eqref{eq:app222}, we have
\bea
\label{eq:BB1}
J_{m}(x) = \frac{1}{m}\left((\theta^W-\theta)J_{m-2}(x)+xJ_{m-1}(x)\right),\quad m\geq 2,
\eea
with $J_0(x)=1$ and $J_1(x)=x$.
\hfill
\end{proof}
\par If we introduce a formal small quantity $\varepsilon$ 
and assume $\rho = \rho_0(1+\bar{\rho}), m_{\boldsymbol{0}}=\rho^W=\rho_0(1+\bar{\rho}^W), 
u_i = \sqrt{\theta_0}\bar{u}_i, u_i^W=\sqrt{\theta_0}\bar{u}_i^W, \theta=\theta_0(1+\bar{\theta}),
\theta^W = \theta_0(1+\bar{\theta}^W)$, $
m_{\boldsymbol{\alpha}}=\rho_0\theta_0^{\frac{|\boldsymbol{\alpha}|}{2}}\bar{m}_{\boldsymbol{\alpha}}$, where
the variables with a bar are $O(\varepsilon)$, then discarding the higher order small quantities we have
\begin{prop}
	$\bar{m}_{\+e_i}=\bar{u}_i^W-\bar{u}_i.\ \bar{m}_{2\+e_i}=\frac{1}{2}
	(\bar{\theta}^W-\bar{\theta})$. $\bar{m}_{\+\alpha}=0$ when
	$\+\alpha\neq\+0,\+e_i,2\+e_i$. And
	\begin{eqnarray}\label{eq:rhow}
	S(0,0)(\bar{\rho}^W-\bar{\rho})=
\sum_{\beta_2=2,\ \text{even}}^{M} 
S(0,\beta_2)(\bar{f}_{\beta_2 \boldsymbol{e}_2}-\bar{m}_{\beta_2 \boldsymbol{e}_2}).
\end{eqnarray}
\end{prop}
\begin{proof}
	Set $\boldsymbol{\alpha}=(0,0,0)$ in \eqref{eq:lbc-0317} and immediately we have \eqref{eq:rhow}.
	Since $J_1(x)=x$ and $J_2(x)=\frac{1}{2}\left(
	\theta_0(\bar{\theta}^W-\bar{\theta})+x^2\right)$, from \eqref{eq:maJ} we have
	$\bar{m}_{\+e_i}=\bar{u}_i^W-\bar{u}_i,\ \bar{m}_{2\+e_i}=\frac{1}{2}
	(\bar{\theta}^W-\bar{\theta})$. Further from \eqref{eq:BB1} we can induce that 
	$J_m(x)=O(\varepsilon^{\lceil\frac{m}{2}\rceil})$ if $x=O(\varepsilon)$ and $\theta^W-\theta=O(\varepsilon)$.
	So when $\+\alpha\neq\+0,\+e_i,2\+e_i$, from \eqref{eq:maJ} we have 
	$\bar{m}_{\+\alpha}=o(\varepsilon)$.
	\hfill
\end{proof}

%\end{CJK*}
\end{document}